\documentclass[runningheads]{llncs}

\usepackage{cite}
\usepackage[T1]{fontenc}
\usepackage{amsmath,amssymb}
\usepackage{xcolor,svgcolor}
\usepackage{pifont,yfonts}
\newcommand{\cmark}{\ding{51}}
\newcommand{\xmark}{\ding{55}}
\usepackage{algorithm,algpseudocode}
\usepackage{arydshln}
\usepackage{pgfplots}
\pgfplotsset{width=8cm,compat=1.9}
\usepackage{subcaption}

\usepackage[backgroundcolor=yellow,colorinlistoftodos]{todonotes}

\usepackage{marginnote}

\usepackage{graphicx}
\usepackage{hyperref}
\hypersetup{
    colorlinks=true,
    linkcolor=darkred,
    filecolor=darkgreen,
    urlcolor=blue,
    pdftitle={Communication optimal UPSU},
    }

\usepackage{newfloat}

\DeclareFloatingEnvironment[ fileext = los, within=none , listname = {List of
protocols} , name = Protocol ]{protocol}

\DeclareFloatingEnvironment[ fileext = los , within=none , listname = {List of
functionalities} , name = Functionality  ]{functionality}

\usepackage[capitalize]{cleveref}
\newcommand{\algabbrev}{Alg}
\newcommand{\algname}{{\algabbrev}.}
\newcommand{\algsname}{{\algabbrev}s.}
\newcommand{\Algname}{Algorithm}
\newcommand{\Algsname}{{\Algname}s}
\crefname{algorithm}{\algname}{\algsname}
\Crefname{algorithm}{\Algname}{\Algsname}
\crefname{proposition}{Prop.}{Props.}
\Crefname{proposition}{Proposition}{Propositions}
\crefname{remark}{Rem.}{Rems.}
\Crefname{remark}{Remark}{Remarks}
\crefname{functionality}{Func.}{Funcs.}
\Crefname{functionality}{Functionality}{Functionalities}
\crefname{protocol}{Protocol}{Protocols}
\crefname{definition}{Def.}{Defs.}

\makeatletter
\newcommand\xleftrightarrow[2][]{%
  \ext@arrow 9999{\longleftrightarrowfill@}{#1}{#2}}
\newcommand\longleftrightarrowfill@{%
  \arrowfill@\leftarrow\relbar\rightarrow}
\makeatother

\usepackage{xspace}
\newcommand{\Setup}{\textbf{Setup}\xspace}
\newcommand{\Encode}{\textbf{Encode}\xspace}
\newcommand{\Reduce}{\textbf{Reduce}\xspace}
\newcommand{\Map}{\textbf{Map}\xspace}
\newcommand{\Union}{\textbf{Union}\xspace}

\begin{document}
\title{Optimal Communication Unbalanced Private Set Union}

\author{Jean-Guillaume Dumas\inst{1} \and
 Alexis Galan\inst{1} \and
 Bruno Grenet\inst{1} \and Aude Maignan\inst{1} \and Daniel S. Roche\inst{2}}
 \institute{Univ. Grenoble Alpes, LJK, UMR CNRS 5224, 38000 Grenoble, France \email{\{firstname.lastname\}@univ-grenoble-alpes.fr}
 \and
 United States Naval Academy, Annapolis, Maryland, United States \email{roche@usna.edu}\\
 }

\maketitle              
\begin{abstract}
We present new two-party protocols for the Unbalanced Private Set Union (UPSU) problem.
Here, the Sender holds a set of data points, and the Receiver holds another (possibly much larger) set, and they would like for the Receiver to learn the union of the two sets and nothing else. Furthermore, the Sender's computational cost, along with the communication complexity, should be smaller when the Sender has a smaller set.
While the UPSU problem has numerous applications and has seen considerable recent attention in the literature, our protocols are the first where the Sender's computational cost and communication volume are linear in the size of the Sender's set only, and do not depend on the size of the Receiver's set.
Our constructions combine linearly homomorphic encryption (LHE) with
fully homomorphic encryption (FHE). The first construction uses multi-point polynomial evaluation (MEv) on FHE, and achieves optimal linear cost for the Sender, but has higher quadratic computational cost for the Receiver. In the second construction we explore another trade-off: the Receiver computes fast polynomial Euclidean remainder in FHE while the Sender computes a fast MEv, in LHE only. This reduces the Receiver's cost to quasi-linear, with a modest increase in the computational cost for the Sender.
Preliminary experimental results using HElib indicate that, for example, a Sender holding 1000 elements can complete our first protocol using about 2s of computation time and less than 9MB of communication volume, independently of the Receiver's set size.

\end{abstract}

\section{Introduction}

A private set union (PSU) protocol is a cryptographic protocol
involving two parties. The receiver, denoted $\mathcal{R}$, owns a set
$\textbf{X}$, and the sender, denoted $\mathcal{S}$, owns a set
\textbf{Y}. The desired functionality of such a protocol
% is denoted
% $\mathcal{F}_{PSU}$ and is presented in~\cref{fun:PSU}:
is that the receiver
$\mathcal{R}$ receives only the union $\textbf{X}\cup \textbf{Y}$, and
the sender $\mathcal{S}$ learns nothing.
(Note that this is equivalent to the receiver learning the set difference
$\textbf{Y} \setminus \textbf{X}$.)
The
protocol is parameterized with (upper bounds on) the set sizes
$|\textbf{X}|$ and $|\textbf{Y}|$, which are therefore implicitly
revealed to both parties as well. However, the sender $\mathcal{S}$
learns nothing about the content of \textbf{X} and the receiver
$\mathcal{R}$ learns nothing about
$\textbf{X}\cap\textbf{Y}$.
% \begin{functionality}[ht]
% \caption{$\mathcal{F}_{PSU}$, Private Set Union Functionality}\label{fun:PSU}
% \begin{center}
%
% \begin{tabular}{r c c c l}
%
% \hline \Large{$\mathcal{R}$} &
%
% \vline$\begin{array}{rr}\textbf{X}\rightarrow\\\textbf{X}\cup\textbf{Y}\leftarrow
% \end{array}$&\fbox{{PSU
% Protocol}}&$\begin{array}{ll}&\leftarrow\textbf{Y}\\
% &\end{array}$\vline &\Large{$\mathcal{S}$} \\ \hline
% \end{tabular}
% \end{center}
% \end{functionality}

PSU protocols have been widely studied in the case of {\em balanced} input set size~\cite{BS05,KS05,DBLP:conf/acns/Frikken07,DBLP:conf/acisp/DavidsonC17,DBLP:conf/asiacrypt/KolesnikovRT019,DBLP:conf/pkc/GarimellaMRSS21,Jia+22,Zhang:2023:Usenix:LPSU}, motivated by numerous practical applications such as
%IP blacklist and vulnerability data aggregation, or
disease data collection from hospitals.

Our interest lies in the {\em unbalanced} setting, in particular where
the sender's input set is (quite) smaller than the receiver's. This
setting has received less attention, but we note two recent works on UPSU
\cite{Tu:2023:CCS:UPSU,DBLP:journals/iacr/ZhangCLPHWW24}
which were developed independently from ours.
The motivation here is data
\emph{aggregation} from multiple sources which may have different sizes
and computational capabilities, and where the set intersections may
reveal private relationships, such as IP blacklists \cite{BLAG}.

As an illustrative example, consider a whistleblower
which has some
confidential and compromising data they would like to share with some
institution.
% The anonymity of the whistleblower should be preserved by
% the protocol, so in particular, the institution should only learn the
% \textit{new} data from the protocol --- otherwise,
If the institution
obtains some data that was already in its set, it could leak some
relationship between different whistleblowers and compromise their
anonymity.
Similarly, if the whistleblower learns anything, this could tell them
about the presence (or not) of other previous whistleblowers.
Furthermore, it is expected that one whistleblower's amount of data as well as their
computational resources are smaller than the institution's, so their
computational cost and communication size should be as small as
possible; hence UPSU.
% During the
% communications, an external adversary should not be able to learn
% anything about the compromising data, owned by the institution or by
% the whistleblower, so the first solution is to encrypt the
% conversations.

If $m$ and $n$ denote the respective sizes of the sets owned by the
sender $\mathcal{S}$ and the receiver $\mathcal{R}$, we thus assume $n
\ge m$ and focus on the case where $n \gg m$. The communication
volume, that is, the least number of elements exchanged during the
protocol, must be $\Omega(m)$ asymptotically, since in the worst case, the
whole set of the sender must be sent to the receiver.
(This also explains why the other unbalanced case of $m \ge n$ is not
promising for improvements.)
Also, the optimal arithmetic cost, that is the number of arithmetic
operations performed by each party, must be at least linear in the
size of their set.

%In~\cref{tab:comparison}, we present the asymptotic costs of three PSU
%protocols~\cite{DBLP:conf/acns/Frikken07, DBLP:conf/acisp/DavidsonC17,
%  Zhang:2023:Usenix:LPSU} adapted in the unbalanced setting. We also
%include the two recent PSU protocols focused on the unbalanced
%case~\cite{Tu:2023:CCS:UPSU, DBLP:journals/iacr/ZhangCLPHWW24} and
%developed independently of ours. Up to our knowledge, those are the
%only two protocols whose communication volume is sub-linear in the
%size of the receiver's set.
%In the table, we compare the arithmetic costs for each party and the
%communication volume, but we also distinguish the protocols through
%their number of rounds, their compatibility with the security
%assumptions and, for those that are using fully homomorphic encryption
%(FHE), as ours, the multiplicative depth of their algorithms, as it
%has a huge impact on the practical performance. A green value
%satisfies our goals, while an orange is expected to be improved and a
%red is not appropriate.

\paragraph{Related Work.}
In~\cite{DBLP:conf/acns/Frikken07} and
in~\cite{DBLP:conf/acisp/DavidsonC17}, the receiver's set is
represented respectively with a polynomial and with a Bloom
filter~\cite{DBLP:journals/cacm/Bloom70} (evaluating to zero
exactly in its element's set). Note that for the Bloom filter case,
there is a probability of false-positive, where an element that is not
in the receiver's set evaluates to zero, that has to be controlled
with a supplementary statistical security parameter. The
representation is sent encrypted under a linearly homomorphic
encryption (LHE)~\cite{DBLP:series/sbcs/YiPB14}, and evaluated
homomorphically by the sender in all its element's set. Upon reception
and decryption of the evaluations, the receiver obtains either zero
and learns nothing, or a non-zero value from which it computes the
evaluation point.

In~\cite{Zhang:2023:Usenix:LPSU}, the receiver represents its set
$\textbf{X}$ as a database called an oblivious key-value store
(OKVS)~\cite{DBLP:conf/crypto/GarimellaPRTY21} in which the elements
of $\textbf{X}$ are viewed as keys, all associated to encryptions of a
same secret value. Upon reception, the sender queries the OKVS using
its elements as keys. By design, the sender gets encryptions of the
secret value for elements in the intersection and random values for
elements not in the intersection. Through a sub-protocol, the parties
obliviously compare the sender's values with the receiver's secret,
and the receiver obtains a bit vector where the ones indicate a
match. Finally, the parties perform an oblivious
transfer~\cite{DBLP:conf/latincrypt/ChouO15} in order for the receiver
to learn exactly the elements of $\textbf{Y}$ that are not in
$\textbf{X}$.

This protocol has been recently translated into two versions of an
UPSU~\cite{DBLP:journals/iacr/ZhangCLPHWW24}. Both versions use the
fact that an OKVS can be structured into a sparse and a dense
parts. The sparse part is decodable with low arithmetic cost and
communication, and only the small dense part has to be
communicated. For efficiency, the first version requires both parties
to store their sets into hash tables, with cuckoo hashing for the
sender, in order to reduce the union protocol between sets of sizes
$n$ and $m$ to about $m$ union protocols between sets of sizes
approximately $n/m$ (partitioning).
The second version uses a re-randomizable public-key encryption
(ReRand-PKE) to skip the sub-protocol step
from~\cite{Zhang:2023:Usenix:LPSU}, which compares the secret and the
decoded values.

In~\cite{Tu:2023:CCS:UPSU}, the global idea is to evaluate a
polynomial representing the receiver's set in the sender's elements as
in~\cite{DBLP:conf/acns/Frikken07}, but unlike in the latter, the
receiver performs the evaluation. Indeed in an unbalanced context, the
large polynomial of the receiver cannot be communicated. To do so, the
sender's elements are sent encrypted under FHE, since polynomial
evaluation requires both additions and multiplications. Following the
ideas from a PSI protocol~\cite{DBLP:conf/ccs/ChenHLR18} to reduce the
multiplicative depth, the protocol uses batching, windowing, oblivious
transfer and partitioning with hash tables. However, as mentioned
in~\cite{DBLP:conf/asiacrypt/KolesnikovRT019} and as we show
in~\cref{sec:leak}, the usage of hash tables to partition the sets,
used in~\cite{Tu:2023:CCS:UPSU} and the first protocol
of~\cite{DBLP:journals/iacr/ZhangCLPHWW24}, undesirably leaks some
information on the intersection set $\textbf{X}\cap\textbf{Y}$.

A comparison of the asymptotic complexity bounds of all these protocols and
ours are detailed in~\cref{tab:comparison}. 
%In~\cref{tab:comparison}, we present the asymptotic costs of three PSU
%protocols~\cite{DBLP:conf/acns/Frikken07, DBLP:conf/acisp/DavidsonC17,
%  Zhang:2023:Usenix:LPSU} adapted in the unbalanced setting. We also
%include the two recent PSU protocols focused on the unbalanced
%case~\cite{Tu:2023:CCS:UPSU, DBLP:journals/iacr/ZhangCLPHWW24} and
%developed independently of ours. Up to our knowledge, those are the
%only two protocols whose communication volume is sub-linear in the
%size of the receiver's set.
We compare the arithmetic costs for each party and the
communication volume, but we also distinguish the protocols through
their number of rounds, their compatibility with the security
assumptions and, for those that are using fully homomorphic encryption
(FHE), as ours, the multiplicative depth of their algorithms, as it
has a huge impact on the practical performance. A green value
satisfies our goals, while an orange is expected to be improved and a
red is not appropriate.

\begin{table*}[h]
\caption{Comparison of (U)PSU protocols, where the receiver $\mathcal{R}$ has a set of size $n$ and the sender $\mathcal{S}$ a set of size $m$, with $n>{m}$ }\label{tab:comparison}
\centering
\begin{tabular}{|c||c|c|c|c|c|c|}
\hline
\textbf{Protocol}&Cost for $\mathcal{R}$& Cost for $\mathcal{S}$& Comm. Vol. & $\#$ rounds & Depth &Security\\
\hline
\cite{DBLP:conf/acns/Frikken07}&$O(n^{1+\epsilon})$&\textcolor{red}{$O(nm)$}&\textcolor{red}{$O(n)$}&2&&\textcolor{darkgreen}{\cmark} \\
\hline
\cite{DBLP:conf/acisp/DavidsonC17}&$O(n)$&\textcolor{orange}{$O(m\log n)$}&\textcolor{red}{$O(n)$}&2&&\textcolor{darkgreen}{\cmark}\\
\hline
\cite{Zhang:2023:Usenix:LPSU}&$O(n)$&\textcolor{orange}{$O(m\log n)$}&\textcolor{red}{$O(n)$}&$\geq 1$+OT&&\textcolor{darkgreen}{\cmark}\\
\hline
\cite{Tu:2023:CCS:UPSU}&$O(n)$&\textcolor{orange}{$O(m\log n)$}&\textcolor{orange}{$O(m\log n)$}&$\geq 4$+OT&\textcolor{darkgreen}{$\leq \log \log (n/m)$}&\textcolor{red}{\xmark}\\
\hline
\cite[$\mathsf{PSU}_{\mathsf{op}}$]{DBLP:journals/iacr/ZhangCLPHWW24}&$O(n)$&\textcolor{orange}{$O(m\log n)$}&\textcolor{orange}{$O(m\log n)$}&$\geq 6$+OT& &\textcolor{red}{\xmark}\\
\hline
\cite[$\mathsf{PSU}_{\mathsf{pk}}$]{DBLP:journals/iacr/ZhangCLPHWW24}&$O(n)$&\textcolor{orange}{$O(m\log n)$}&\textcolor{orange}{$O(m\log n)$}&$\geq 2$+OT&&\textcolor{darkgreen}{\cmark}\\
\hline
\cref{pro:para/batch}&$O(mn)$&\textcolor{darkgreen}{$O(m)$}&\textcolor{darkgreen}{$O(m)$}&3&\textcolor{orange}{$\log n+1$}&\textcolor{darkgreen}{\cmark}\\
\hline
\cref{pro:modF}&$O(n^{1+\epsilon})$&\textcolor{darkgreen}{$O(m^{1+\epsilon})$}&\textcolor{darkgreen}{$O(m)$}&3&\textcolor{orange}{$2\log (n/m)+1$}&\textcolor{darkgreen}{\cmark}\\
\hline
\end{tabular}
\end{table*}

\paragraph{Our contributions.}
We present two new UPSU protocols (the second one with a variant),
proven secure under the honest-but-curious adversary model. 

All our
protocols have a communication volume linear in the size of the
sender's set and independent of the size of the receiver's set.

We also have an optimal number of rounds and a sender's arithmetic cost
which is independent of the receiver's set size.

Our protocols combine two encryption schemes, namely a
linearly homomorphic one and a fully homomorphic one.

We then need to introduce several homomorphic algorithms on
polynomials, which can be of independent interest.
Our main tool is to use a polynomial representing the receiver's set
(its roots are the receiver's elements) and evaluate it
homomorphically in each of the sender's elements, but without
communicating the whole polynomial. 

We optimize here a trade-off
between communication volume and FHE multiplicative depth.
For security and correctness purposes, we need the FHE scheme to
allow exact polynomial evaluation, so we need the plaintext space to
be an integral subdomain of the rationals. 

The BGV
cryptosystem~\cite{DBLP:journals/toct/BrakerskiGV14} can for instance
satisfy this condition as its plaintext space is a finite field.
\begin{itemize}
\item Our first protocol is built on fully homorphic multi-point evaluation,
  with homomorphic scalar products, low multiplicative depth and large
  parallelism.
  We use a few FHE optimizations, namely batching and
  modulus switching, in order to reduce the time computation and the
  communication volume. We propose an implementation of this protocol
  with low communication, using the
  HElib\footnote{\url{https://github.com/homenc/HElib}} instantiation
  of BGV.
\item Our second protocol relies on efficient fully homomorphic
  Euclidean remainder and only linearly homomorphic multi-point evaluation.
  This drastically reduces the arithmetic cost for the
  receiver, but increases the multiplicative depth and the sender's
  cost. Some other trade-offs can be considered and, as a variant, we
  for instance consider a third protocol, with a slightly worse depth,
  but a better sender's cost.
\end{itemize}

\paragraph{Outline.}
In~\cref{sec1}, we introduce the adversary model and the formal
security definitions. We then propose in~\cref{sec:leak} a privacy
attack on the partitioned constructions with hash tables
of~\cite{Tu:2023:CCS:UPSU,DBLP:journals/iacr/ZhangCLPHWW24}.
\cref{sec:blocks} contains the linear and fully homomorphic
encryption schemes formalization with some practical aspects.
From this, in~\cref{sec:UPSU}, we present our optimal communication
protocol, using homomorphic batched scalar multi-point evaluation. We
also present an implementation with computational timings and
communication volume. Finally, in~\cref{sec:variants} we present a
protocol (and a variant), using efficient homomorphic algorithms on
polynomials, that improve on the asymptotic complexity bounds.
%-------------------------------------------------------------------------------
%-------------------------------------------------------------------------------
%-------------------------------------------------------------------------------
\section{Adversary Model and UPSU Security Definition}\label{sec1}
%-------------------------------------------------------------------------------
\subsubsection{Honest-But-Curious Adversaries.}
Our protocols are secure under the honest-but-curious adversary model,
where the participants must follow the protocol but try to learn as
much additional information as possible. The security proofs of our
protocols, presented in~\cref{app:secuUPSU}, are by simulation,
following the framework of~\cite{DBLP:books/sp/17/Lindell17},
where a probabilistic polynomial time (PPT) simulator can generate
computationally indistinguishable
transcripts~\cite{DBLP:books/cu/Goldreich2001}.
\subsubsection{Unbalanced Private Set Union Definition.}
We propose a formal definition of an UPSU protocol divided into five algorithms, namely \Setup, \Encode, \Reduce, \Map and \Union. For a sender $\mathcal{S}$ that owns a set $\textbf{Y}$ and a receiver $\mathcal{R}$ that owns a set $\textbf{X}$:
\begin{itemize}
\item $\lbrace keys_\mathcal{R},keys_\mathcal{S}\rbrace\leftarrow\Setup(\kappa,\lambda)$: On input of a computational security parameter $\kappa$ and optionally a statistical security parameter $\lambda$, set up a context (encryption schemes, hash functions, ...) with respect to $\kappa$ and $\lambda$, and  outputs $keys_\mathcal{R}$ to the receiver and $keys_\mathcal{S}$ to the sender;
\item $E_\textbf{Y}\leftarrow\Encode(\textbf{Y},keys_\mathcal{S})$: Given the set \textbf{Y} and $keys_\mathcal{S}$, outputs $E_\textbf{Y}$ to the receiver, an encoding of the set \textbf{Y};
\item $R_{\textbf{X}|E_\textbf{Y}}\leftarrow\Reduce(\textbf{X},E_\textbf{Y},keys_\mathcal{R})$: As input, takes the set \textbf{X}, $keys_\mathcal{R}$ and $E_\textbf{Y}$. Outputs $R_{\textbf{X}|E_\textbf{Y}}$ to the sender, an encoding of the set \textbf{X}, reduced in size depending on $E_\textbf{Y}$;
\item $M_{Y|R_{\textbf{X}}}\leftarrow\Map(\textbf{Y},R_{\textbf{X}|E_\textbf{Y}},keys_\mathcal{S})$: On input of a set $\textbf{Y}$, an encoding $R_{\textbf{X}|E_\textbf{Y}}$ and $keys_\mathcal{S}$, outputs to the receiver an encoded data set $M_{Y|R_{\textbf{X}}}$ representing the set $\textbf{Y}\setminus\textbf{X}$, depending on $\textbf{Y}$ and $R_{\textbf{X}|E_\textbf{Y}}$;
\item $\textbf{Z}\leftarrow\Union(\textbf{X},M_{Y|R_{\textbf{X}}},keys_\mathcal{R})$: On input of the set \textbf{X}, an encoded data set $M_{Y|R_{\textbf{X}}}$ and $keys_\mathcal{R}$, outputs the union set $\textbf{Z}=\textbf{X}\cup\textbf{Y}$ to the receiver.
\end{itemize}
\begin{definition}
(\Setup, \Encode, \Reduce, \Map, \Union) is an unbalanced private set union scheme if it satisfies the following three properties:
\begin{enumerate}
\item \textbf{Correctness.} For security parameters $\kappa$ and $\lambda$ and any sets $\textbf{X},\textbf{Y}$, for
\begin{align*}
\lbrace keys_\mathcal{R}, keys_\mathcal{S}\rbrace&\leftarrow\Setup(\kappa,\lambda),\\
E_\textbf{Y}&\leftarrow\Encode(\textbf{Y},keys_\mathcal{S}),\\
R_{\textbf{X}|E_\textbf{Y}}&\leftarrow\Reduce(\textbf{X},E_\textbf{Y},keys_\mathcal{R}),
\end{align*}
the scheme is correct if
\begin{equation}
\Union(\textbf{X},\Map(\textbf{Y},R_{\textbf{X}|E_\textbf{Y}},keys_\mathcal{S}),keys_\mathcal{R})=\textbf{X}\cup\textbf{Y}.
\end{equation}
\item \textbf{Privacy.} The scheme is secure under the honest-but-curious adversary model. In particular, the receiver learns $\textbf{X}\cup\textbf{Y}$ but nothing about $\textbf{X}\cap\textbf{Y}$, and the sender learns nothing.
\item \textbf{Unbalanced efficiency.} For input sets $\textbf{X}$ for the receiver and $\textbf{Y}$ for the sender, if $|\textbf{Y}|=o(|\textbf{X}|)$, then the total communication volume of the scheme, as well as the sender's arithmetic cost, are $o(|\textbf{X}|)$.
\end{enumerate}
\end{definition}
\begin{remark}
Note that a non-unbalanced PSU protocol can be described with those
algorithms, considering that \Encode outputs $\emptyset$. However,
generally such a protocol will not satisfy the \textbf{Unbalanced
  efficiency} of the definition, in particular because usually the
receiver sends the first message that has a size proportional to its
set size, even if $|\textbf{Y}|=o(|\textbf{X}|)$.
This justifies the fact that the sender has to send the first message
in \Encode. Then, the receiver uses that message to \Reduce the
representation of its set to a smaller size. A third message is then
mandatory for the receiver to get the final information about the
union. Up to our knowledge, only~\cite{Tu:2023:CCS:UPSU}
and~\cite{DBLP:journals/iacr/ZhangCLPHWW24} are focused on the
unbalanced situation, and their protocols also fit this definition.
\end{remark}
%-------------------------------------------------------------------------------
%-------------------------------------------------------------------------------
%-------------------------------------------------------------------------------
\section{Polynomial-time Attack on the Partitioning of
  \cite{Tu:2023:CCS:UPSU,DBLP:journals/iacr/ZhangCLPHWW24}}\label{sec:leak}
We show here that there is a leaky Construction
in~\cite{Tu:2023:CCS:UPSU} that is reused in the first protocol
of~\cite{DBLP:journals/iacr/ZhangCLPHWW24} ($\mathsf{PSU}_{\mathsf{op}}$).
We do not delve into the full constructions, as the leaks actually occur
already in the setup phase for both of them.
In both protocols indeed, the sender arranges its small set in a
cuckoo hash table~\cite{DBLP:journals/ipl/DevroyeM03}, and the
receiver uses the same hash functions to arrange its large set in a
simple hash table~\cite{DBLP:conf/uss/Pinkas0Z14}.

A required property of both protocols is that the cuckoo hashing must
limit the hash table to have at most one of the sender's elements per bin.
Their protocols, taking as inputs the two sets can now be partitioned
into smaller protocols for each bin, taking as inputs the unique sender's
element and the reduced amount of receiver's elements contained in that bin.
This trick is widely used in private set intersection (PSI) protocols
for efficiency purposes but, as already mentioned
in~\cite{DBLP:conf/asiacrypt/KolesnikovRT019}, it cannot be used
directly in PSU protocols: there it can leak some clues on the set
intersection.

We further show in this section that the leaks increase with the
unbalancedness.
It is therefore always dangerous to use this trick in UPSU protocols.
In the following, the sender $\mathcal{S}$ owns a set \textbf{Y} of $m$
elements and the receiver $\mathcal{R}$ owns a set \textbf{X} of $n$
elements.
\subsection{Hashing Table Procedure.} For a statistical security
parameter $\lambda$, $\mathcal{S}$ selects publicly $i$ hash functions
$h_j:\lbrace 0,1\rbrace^{*}\longrightarrow [k]$, where
$k=k(\lambda)\approx m$ and $i=3$ for cost efficiency. Those hash
functions are chosen such that, with a failure probability bellow
$2^{-\lambda}$, the set \textbf{Y} can fit in a cuckoo hash table of
$k$ bins, without stash.
The cuckoo hashing of \textbf{Y} with $h_1,h_2,h_3$ then places each
element $y\in\textbf{Y}$ in exactly one bin between $h_1(y)$, $h_2(y)$
or $h_3(y)$, such that at the end of the procedure, each bin contains
at most $1$ element. On the other side, $\mathcal{R}$ hashes its set
\textbf{X} with a simple hash table, and each element
$x\in\textbf{X}$ is sent in all three bins $h_1(x)$, $h_2(x)$ and
$h_3(x)$.
\subsection{Leakage on the Intersection.}
By definition, a (U)PSU protocol should not leak any clue about the
intersection set $\textbf{X}\cap\textbf{Y}$ to the receiver.
Now, let $\mathcal{R}$ use the three hash functions $h_1,h_2,h_3$
given by $\mathcal{S}$, to hash \textbf{X} in the simple hash table
$\textbf{X}_S$; this hash table thus contains $3n$ elements in $k$
bins.
If there are $4$ distinct elements $x_1,x_2,x_3,x_4\in\textbf{X}$ in
at most 3 bins $b_1,b_2,b_3\in[k]$, then the receiver learns that
$\lbrace x_1,x_2,x_3,x_4\rbrace\not\subset \textbf{Y}$: indeed, the
hash functions are chosen by $\mathcal{S}$ such that each bin in the
sender's cuckoo hash table contains at most 1 element of \textbf{Y}.
Therefore, if $x_1,x_2,x_3,x_4\in\textbf{Y}$, then those elements
would have been sent in $4$ different bins, not $3$.
\subsection{Modelization.} We consider that the three hash
functions are independent and send any element uniformly at random in
$[k]$.
Let $(b_1,b_2,b_3)\in [k]^3$ be a bin triplet, with distinct
$b_1,b_2,b_3$.
The probability that an element $x$ is sent in
bins $\lbrace b_1,b_2,b_3\rbrace$ using one hash function is thus
$\frac{3}{k}$. This implies, with independence, that an element $x$ is
sent, resp., in bins $\lbrace b_1,b_2,b_3\rbrace$, resp. with $h_1$,
$h_2$ and $h_3$, with probability $p:=\left(\frac{3}{k}\right)^3$.
If $n$ elements are all sent in bins using $h_1$, $h_2$ and $h_3$, we
want to compute the probability that at least 4 of these elements are
sent in $\{b_1,b_2,b_3\}$. We thus model this with a random
variable~$V$ following a binomial distribution $\mathcal{B}(n,p)$.
Now $\mathbb{P}(V\geq 4)$ represents the sought probability with:
\begin{equation}\label{eq:pv4}
\mathbb{P}(V\geq 4)=1-\sum\limits_{i=0}^3\binom{n}{i} p^i(1-p)^{n-i}.
\end{equation}
\subsection{Probability of leakage.} In both papers, the
statistical security parameter is $\lambda=40$ and the size of the
hash table is $k\approx m+\log m$. In~\cref{tab:leak}, we
instantiate~\cref{eq:pv4} for some realistic unbalanced
parameters~$n,m$.
\begin{table}[ht]
\centering
\caption{Leaky situation probability lower bounds when $n>m$}\label{tab:leak}
\begin{tabular}{|c|c||c|}
\hline
$n$&$m$&$\mathbb{P}(V\geq 4)$\\
\hline
$2^{20}$&$2^{10}$&$\geq 2^{-26}$\\
\hline
$2^{10}$&$10$&$\geq 2^{-0,0138}\approx 99.05\%$\\
\hline
$2^{20}$&$10$&$\geq 1-10^{(-4492)}$\\
\hline
\end{tabular}
\end{table}
\Cref{tab:leak} shows that this crude lower bound on the probability
of leak is actually way larger than the statistical security parameter,
already for a single triplet of bins (while there are in fact
$n\choose{3}$, not independent, triplets).
\subsection{Polynomial-time Attack}
To prevent these leaks, one could increase the parameter $k$, \emph{viz.} $k>3\sqrt[3]{2^\lambda n}$, but this
brings an overhead in communication volume and arithmetic cost for the
sender that is no longer sustainable in the unbalanced context.
Note also that the leaky situation occurs in fact with probability 1
if $n>3k^3+1$: by the pigeonhole principle there must then be at least
one set of $3$ bins with $4$ distinct elements of $\textbf{X}$.
This is the root for an always possible polynomial-time
attack: in this case, an honest-but-curious attacker with an input set
\textbf{X}, of size $n$, knowing the three hash functions
$h_1,h_2,h_3:\lbrace 0,1\rbrace^{*}\longrightarrow [k]$, can compute
its simple hash table $\textbf{X}_S$. The attacker can check if $4$ of
distinct elements fall in a set of $3$ bins.
Otherwise, the attacker just adds further distinct elements of its
choice in the hash table, until there is a leak (at most $3k^3+2-n$
new known elements have to be added).
Thus partionning is always subject to this attack, requiring only a
polynomial number of operations.

In the following, we therefore propose proven secure protocols that do
not make use of partionning.
%-------------------------------------------------------------------------------
%-------------------------------------------------------------------------------
%-------------------------------------------------------------------------------
\section{Cryptographic Tools: Homomorphic Schemes}\label{sec:blocks}
\subsection{Linearly Homomorphic Encryption Scheme}
\textit{Notation:} In the following, we denote by $\widehat{x}$ a
value encrypted using linearly homomorphic encryption (LHE).

A LHE is a semantically secure public-key encryption scheme such that
for $(pk_L,sk_L)\leftarrow\textbf{L.Setup}(\kappa)$, a key pair for a
security parameter $\kappa$, together with
$\widehat{m_1}\leftarrow\textbf{L.E}_{pk_L}(m_1)$ and
$\widehat{m_2}\leftarrow\textbf{L.E}_{pk_L}(m_2)$, two encryptions,
their decryption $\textbf{L.D}$ must satisfy:
\begin{itemize}
\item \textit{homomorphic addition $+_L$: } $\textbf{L.D}_{sk_L}(\widehat{m_1}+_L \widehat{m_2})=m_1+m_2$.
\item \textit{cleartext-ciphertext product $\ltimes_L$: } $\textbf{L.D}_{sk_L}(m_1\ltimes_L \widehat{m_2})=m_1m_2$.
\end{itemize}
For the remainder, we assume that the chosen LHE scheme satisfies
IND-CPA security.
\begin{remark}\label{rem:Lvectmat}
We extend naturally the encryption and decryption algorithms for a LHE
to allow matrix or polynomials (seen as a vector of its coefficients)
as inputs. This allows for instance to extend $+_L$ and $\ltimes_L$ to
matrix or polynomial inputs. In particular, we can compute the
homomorphic polynomial product between a cleartext polynomial $A$ and
a ciphertext polynomial $\widehat{C}$.
\end{remark}
%-------------------------------------------------------------------------------
\subsection{Fully Homomorphic Encryption Scheme}
\textit{Notation:} In the following, we denote by  $\widetilde{x}$ a
value encrypted using fully homomorphic encryption (FHE).

A FHE is a semantically secure public-key encryption scheme such that
for $(pk_F,sk_F)\leftarrow\textbf{F.Setup}(\kappa)$, a key pair for a
security parameter $\kappa$,  together with
$\widetilde{m_1}\leftarrow\textbf{F.E}_{pk_F}(m_1)$ and
$\widetilde{m_2}\leftarrow\textbf{F.E}_{pk_F}(m_2)$, two encryptions,
their decryption $\textbf{F.D}$ must satisfy:
\begin{itemize}
\item \textit{homomorphic addition $+_F$: } $\textbf{F.D}_{sk_F}(\widetilde{m_1}+_F \widetilde{m_2})=m_1+m_2$.
\item \textit{cleartext-ciphertext product $\ltimes_F$: } $\textbf{L.D}_{sk_F}(m_1\ltimes_F \widetilde{m_2})=m_1m_2$.
\item \textit{homomorphic product $\times_F$: } $\textbf{F.D}_{sk_F}(\widetilde{m_1}\times_F \widetilde{m_2})=m_1m_2$.
\end{itemize}
\begin{remark}
Similarly to the LHE, we extend naturally the scheme for matrix and
polynomial inputs, and we assume for the purpose of our proofs
that the FHE is IND-CPA secure.
\end{remark}
\subsection{Practical Tools in FHE.}
In practice, actual FHE schemes, like the BGV
scheme~\cite{DBLP:journals/toct/BrakerskiGV14}, are constructed on top
of a leveled fully homomorphic encryption scheme, allowing a bounded
multiplicative depth algorithms.
The latter is extended to arbitrary depth using a
\textit{bootstrapping} procedure. A message is encrypted with a random
noise that grows at each homomorphic operation.
If the noise becomes too large, the ciphertext is no longer
decipherable.
Before that, a \textit{modulus switching} procedure can be performed
(in particular, after each homomorphic product), reducing the noise
(and the size of the ciphertext).
When the ciphertext size is even too small for another modulus
switching, then only a bootstrapping is required (increasing back the
size of the ciphertext, while regaining a small noise).
Bootstrapping being usually quite expensive, it means that the
multiplicative depth of the protocols has to be controlled.
We further introduce two tools that we use in practice in our
protocols. The first one is a \textit{noise flooding} which, together with a
\textit{shortening}, allows to ensure {\em circuit privacy}, while
incidentally reducing the communication volume.
The second tool is \textit{batching}, that allows
single-instruction multiple-data (SIMD) homomorphic operations.
\subsubsection{Noise Flooding and Shortening}
In some implementations of FHE, ciphertexts size and noise could
contain some information about the homomorphic operation performed.
To prevent this, and thus ensure circuit privacy, some schemes often
need another algorithm, called \textbf{flood}, that, given an FHE
ciphertext $\widetilde{y}$, performs a noise flooding\footnote{For
  instance, by adding a random noise.} together with as many modulus
switching as possible (to preserve the decryption's correctness).
The resulting ciphertext encrypts the same cleartext, but can no more
be multiplied homomorphically without a bootstrap.
Now, for $f,g$ two arithmetic circuits, let $a,b,y$ be such that
$f(a)=g(b)=y$ and let $\widetilde{a},\widetilde{b}$ be their
respective FHE encryptions, the decryption key owner should then not be
able to distinguish $\textbf{flood}(f(\widetilde{a}))$ from $\textbf{flood}(g(\widetilde{b}))$.

\subsubsection{Batching}
From the construction of cleartext and ciphertext spaces in FHE
schemes, it is often possible to batch together several cleartexts so
that encrypting them will result in a single ciphertext.
We will denote by $\langle y_i\rangle_{i\in[m]}$ a batch of $m$
cleartexts. For
$\widetilde{\textbf{y}}\leftarrow\textbf{F.E}_{pk_F}(\langle
y_i\rangle_{i\in[m]})$ an encryption of batched cleartexts and $f$ a
circuit, the batching correctness implies that
$f(\widetilde{\textbf{y}})$ is an encryption of $\langle
f(y_i)\rangle_{i\in[m]}$.
%\remark Those tools are not affecting the asymptotic of our protocol in order to stay consistent in the comparison with other protocols. For that reason, a ciphertext encrypting $m$ batched cleartexts is counted as $O(m)$ communications, performing an homomorphic operation on it is counted as performing $m$ times this operation and the algorithm \textbf{flood} does not change the asymptotic communication volume neither.
%-------------------------------------------------------------------------------
%-------------------------------------------------------------------------------
%-------------------------------------------------------------------------------
\section{Optimal Communication Volume, Low Depth, Batchable and Parallelizable UPSU Protocol} \label{sec:UPSU}
In our protocol, we represent the set $\mathbf{X}$ owned by the
receiver as the polynomial $P_{\mathcal{R}}(Z) =
\prod_{x\in\mathbf{X}} (Z-x)$ and we evaluate $P_{\mathcal{R}}$  on
each element owned by the sender. The receiver must learn only the
elements of $\mathbf{Y}$ that are not a root of $P_{\mathcal{R}}$, and
the evaluation must be performed homomorphically. The idea to keep the
communication volume proportional to the size $|\mathbf{Y}|$ of the
sender's set, is that the sender first sends its elements, encrypted. The receiver performs the evaluation homomorphically on its polynomial in clear. This requires both additions and multiplications, whence the need for a FHE scheme. For efficiency purposes, we introduce a homomorphic polynomial evaluation algorithm that is highly parallelizable, uses properly the batch and has a low multiplicative depth. Finally, to minimize the number of rounds and avoid the expensive usage of FHE when it is not required, we make a transition from the FHE to the LHE scheme and conclude our protocol similarly to~\cite{DBLP:conf/acns/Frikken07}.

\subsection{Fully Homomorphic Batched Scalar Multi-point Evaluation }
For a pair of FHE keys $(pk_F,sk_F)$, a homomorphic batched scalar multi-point evaluation algorithm, denoted \textbf{F.BSMEv}, is an algorithm that, given a batched ciphertext $\widetilde{\textbf{y}}\leftarrow\textbf{F.E}_{pk_F}(\langle y_i\rangle_{i\in[m]})$ and a polynomial $P$ in clear, satisfies
\begin{equation}
\textbf{F.D}_{sk_F}(\textbf{F.BSMEv}(P,\widetilde{\textbf{y}}))=\langle P(y_i)\rangle_{i\in[m]}.
\end{equation}
This algorithm evaluates a clear polynomial in many encrypted evaluation points. Its main goal is to be efficient in practice. This requires to use properly the batching, to have a low multiplicative depth, to lessen the number of homomorphic products and to be as parallelizable as possible.
\begin{proposition}\label{prop:BSMEv}
Let $P$ be a clear polynomial of degree $n$ given as a product of $\sqrt{n}$ polynomials of degrees $\sqrt{n}$, and let $\widetilde{\textbf{y}}\leftarrow\textbf{F.E}_{pk_F}(\langle y_i\rangle_{i\in[m]})$ be a batched ciphertext corresponding to $m$ plaintexts. Then $\textbf{F.BSMEv}(P,\widetilde{\textbf{y}})$ can be computed in $O(mn)$ arithmetic operations and a depth $\lceil \log n\rceil+1$.
\end{proposition}
\begin{proof}
The algorithm is built in three steps. First, compute homomorphically the encryption of $y_i^j$, for $i\in[m]$ and $j\in[\sqrt{n}]$. Using batching, this costs $m\sqrt{n}$ homomorphic products $\times_F$, in depth $\log \sqrt{n}$. The result is an encrypted vector $\overrightarrow{\widetilde{\textbf{v}}}$. Let $P=P_1\times\cdots\times P_{\sqrt{n}}$, with $\deg P_j=\sqrt{n}$. Then, write %$P=\prod_{j\in\sqrt{n}}P_j$ as a product of polynomials of degrees $\sqrt{n}$, and
each $P_j$ as a vector of coefficients and %With a homomorphic scalar product between the vectors of coefficients and $\overrightarrow{\widetilde{\textbf{v}}}$, we obtain homomorphically $\lbrace \widetilde{P_j(\textbf{y})}\rbrace_{j\in[\sqrt{n}]}$, where
compute the encryptions $\widetilde{P_j(\textbf{y})}$ of $\langle P_j(y_i)\rangle_{i\in[m]}$, for $j\in[\sqrt{n}]$, using homomorphic inner products between the vectors of coefficients and $\overrightarrow{\widetilde{\textbf{v}}}$.
%is an encryption of the batched evaluations $\langle P_j(y_i)\rangle_{i\in[m]}$. It
This requires $nm$ cleartext-ciphertext products $\ltimes_F$, in multiplicative depth $1$. Finally, reconstruct homomorphically the encryption of $\langle P(y_i)\rangle_{i\in[m]}$ from $\lbrace \widetilde{P_j(\textbf{y})}\rbrace_{j\in[\sqrt{n}]}$ with a (homomorphic) binary multiplicative tree. This requires $m\sqrt{n}$ homomorphic products $\times_F$ in depth $\log \sqrt{n}$. %Remark that each step, especially the second one, are parallelizable.
\end{proof}
\subsection{Formalization of the protocol} %Optimal Communication Volume, Low Depth, Batchable and Parallelisable UPSU protocol}
Let $\{x_i\}_{i\in[n]}$, $\lbrace y_i\rbrace_{i\in[m]}\subset\mathbb{M}$, $n\ge m$, be the sets of the receiver $\mathcal{R}$ and of the sender $\mathcal{S}$, respectively. %and $\lbrace x_i\rbrace_{i\in[n]}\subset\mathbb{M}$, be the receiver's $\mathcal{R}$, with $n\geq m$.
For the sake of simplicity, we assume that $\mathbb{M}$ is a finite field, and that it is the (common) plaintext space of an FHE and an LHE. This assumption is met in our implementations, and we propose in~\cref{subsec:compati} some slight changes to keep the correctness of our protocol under other assumptions. The respective ciphertext spaces for LHE and FHE are denoted $\mathbb{E}_L$ and $\mathbb{E}_F$.
Formally, our protocol is built with the algorithms \Setup, \Encode,
\Reduce, \Map and \Union respectively presented
in~\cref{algo:Setup1,algo:Encode1,algo:Reduce1,algo:Map1,algo:Union1}.
A more visual version is presented in~\cref{pro:para/batch}.

\begin{theorem}\label{thm:UPSU1}
The protocol built with the algorithms \Setup, \Encode, \Reduce, \Map
and \Union
(\cref{algo:Setup1,algo:Encode1,algo:Reduce1,algo:Map1,algo:Union1})
is a secure unbalanced private set union scheme under the
honest-but-curious adversary model and computes the set union with the
asymptotic complexity bounds presented in~\cref{tab:costUPSU1}
\begin{table}[ht]
\centering
\caption{Cost analysis of~\cref{pro:para/batch} for $n> m$}\label{tab:costUPSU1}
\begin{tabular}{l||c|c|c|c|}

Algorithm&Ar. Cost for $\mathcal{R}$&Ar. Cost for $\mathcal{S}$&
Comm. Vol.& Depth\\
\hline\hline
\Setup &$O(1)$&$O(1)$&$O(1)$&\\
\hdashline
\Encode &$O(1)$&\boldmath$O(m)$&\boldmath$O(m)$&\\
\hdashline
\Reduce &\boldmath$O(mn)$&$O(1)$&\boldmath$O(m)$&\boldmath$\lceil\log  n\rceil+1$\\
\hdashline
\Map &$O(1)$&\boldmath$O(m)$&\boldmath$O(m)$&\\
\hdashline
\Union &$O(m)$&$O(1)$&$O(1)$&\\
\hline\hline
\textbf{Total}&$O(mn)$&$O(m)$&$O(m)$&$\lceil\log n\rceil+1$\\

\end{tabular}
\end{table}.
\end{theorem}
\begin{proof}
The correctness of the scheme relies on the fact that an element $y$ owned by the sender has to be added to the receiver's set \textbf{X} if, and only if, $P_{\mathcal{R}}(y)\neq 0$. The correctness of the evaluations is implied by the correctness of the FHE and LHE encryption schemes, as well as the fact that both scheme share the same field as plaintext space. The security proof, using simulation, is presented in~\cref{app:secuUPSU}.
\end{proof}

\begin{algorithm}[htbp]
\caption{$\Setup(\kappa)$}\label{algo:Setup1}
\begin{flushleft}
\textbf{Input:} A security parameter $\kappa$.\\
    \textbf{Outputs:} Pairs of LHE keys $(pk_L,sk_L)$ and FHE keys $(pk_F,sk_F)$, both with $\kappa$ bits of security. The secret key $sk_L$ is owned by the receiver, and $sk_F$ is owned by the sender.
%\textbf{Output:} A pair of FHE keys $(pk_F,sk_F)$ and a LHE public key $pk_L$, both schemes with atleast $\kappa$-bit security.
\end{flushleft}
\begin{algorithmic}[1]
\State $\mathcal{R}$: compute $(pk_L,sk_L)\leftarrow \textbf{L.Setup}(\kappa)$ and send $pk_L$ to $\mathcal{S}$;
\State $\mathcal{S}$: compute $(pk_F,sk_F)\leftarrow \textbf{F.Setup}(\kappa)$ and send $pk_F$ to $\mathcal{R}$;
\State $\mathcal{R}$: \Return $keys_\mathcal{R}\leftarrow\lbrace(pk_L,sk_L),pk_F\rbrace$;
\State $\mathcal{S}$: \Return $keys_\mathcal{S}\leftarrow\lbrace(pk_F,sk_F),pk_L\rbrace$;
\end{algorithmic}
\end{algorithm}

\begin{algorithm}[htbp]
\caption{$\Encode(\textbf{Y},keys_\mathcal{S})$}\label{algo:Encode1}
\begin{flushleft} \textbf{Input:} A set of plaintexts $\textbf{Y}=\lbrace y_i\rbrace_{i\in[m]}\subset\mathbb{M}$ and $keys_\mathcal{S}=\lbrace(pk_F,sk_F),pk_L\rbrace$.\\
\textbf{Output:} A FHE ciphertext $\widetilde{\textbf{y}}\in{\mathbb{E}_F}$, such that $\textbf{F.D}_{sk_F}(\widetilde{\textbf{y}})=\langle y_i\rangle_{i\in[m]}$.
\end{flushleft}
\begin{algorithmic}[1]
\State $\mathcal{S}$: compute $\widetilde{\textbf{y}}\leftarrow\textbf{F.E}_{pk_F}(\langle y_i\rangle_{i\in[m]})$ and send $\widetilde{\textbf{y}}$ to $\mathcal{R}$;\Comment{Batch and encrypt}
\State $\mathcal{R}$: \Return $\widetilde{\textbf{y}}$;
\end{algorithmic}
\end{algorithm}

\begin{algorithm}[htbp]
\caption{$\Reduce(\textbf{X},\widetilde{\textbf{y}},keys_\mathcal{R})$}\label{algo:Reduce1}
\begin{flushleft} \textbf{Input:} A set of plaintexts $\textbf{X}=\lbrace x_j\rbrace_{j\in[n]}\subset\mathbb{M}$, a FHE ciphertext $\widetilde{\textbf{y}}\in{\mathbb{E}_F}$ and $keys_\mathcal{R}=\lbrace(pk_L,sk_L),pk_F\rbrace$.\\
\textbf{Output:} A reduced FHE ciphertext $\widetilde{\textbf{h}}\in{\mathbb{E}_F}$, and a set of LHE ciphertexts $\lbrace \widehat{k_i}\rbrace_{i\in[m]}\subset{\mathbb{E}_L}$, such that, for $\textbf{F.D}_{sk_F}(\widetilde{\textbf{y}})=\langle y_i\rangle_{i\in[m]}$, $\textbf{F.D}_{sk_F}(\widetilde{\textbf{h}})=\langle \textbf{L.D}_{sk_L}(\widehat{k_i})+\prod\limits_{j\in[n]} (y_i-x_j)\rangle_{i\in[m]}$.
\end{flushleft}
\begin{algorithmic}[1]
\State $\mathcal{R}$: compute $P_\mathcal{R}\leftarrow \prod\limits_{j\in[n]}(Z-x_j)$;
\ForAll{$i \in [m]$}
\State $\mathcal{R}$: sample $k_i\xleftarrow{\$}\mathbb{M}$ uniformly;
\State $\mathcal{R}$: compute $\widehat{k_i}\leftarrow \textbf{L.E}_{pk_L}(k_i)$;\Comment{Encrypt the masks under LHE}
\EndFor
\State $\mathcal{R}$: compute $\widetilde{\textbf{k}}\leftarrow \textbf{F.E}_{pk_F}(\langle k_i\rangle_{i\in[m]})$;\Comment{Encrypt the batched masks under FHE}
\State $\mathcal{R}$: compute $\widetilde{\textbf{e}}\leftarrow\textbf{F.BSMEv}(P_\mathcal{R},\widetilde{\textbf{y}})$;\Comment{Homom. multi-point evaluation}
\State $\mathcal{R}$: compute $\widetilde{\textbf{h}}\leftarrow\textbf{flood}\left(\widetilde{\textbf{k}}+_F\widetilde{\textbf{e}}\right)$;\Comment{Mask, reduce and flood the evaluations}
\State $\mathcal{R}$: send $\widetilde{\textbf{h}}$ and $\lbrace \widehat{k_i}\rbrace_{i\in[m]}$ to $\mathcal{S}$;
\State $\mathcal{S}$: \Return $\lbrace\widetilde{\textbf{h}},\lbrace \widehat{k_i}\rbrace_{i\in[m]}\rbrace$;
\end{algorithmic}
\end{algorithm}

\begin{algorithm}[htbp]
\caption{$\Map(\textbf{Y},\lbrace\widetilde{\textbf{h}},\lbrace \widehat{k_i}\rbrace_{i\in[m]}\rbrace,keys_\mathcal{S})$}\label{algo:Map1}
\begin{flushleft}
\textbf{Input:} A set of $m$ plaintexts $\textbf{Y}=\lbrace y_i\rbrace_{i\in[m]}\subset\mathbb{M}$, a FHE ciphertext $\widetilde{\textbf{h}}\in{\mathbb{E}_F}$, a LHE ciphertext set $\lbrace \widehat{k_i}\rbrace_{i\in[m]}\subset{\mathbb{E}_L}$, and $keys_\mathcal{S}=\lbrace(pk_F,sk_F),pk_L\rbrace$.\\
\textbf{Output:} A set of LHE ciphertext pairs $\lbrace (\widehat{e_j},\widehat{\eta_j})\rbrace_{j\in[m]}\subset\mathbb{E}_L^2$ such that, for $\textbf{F.D}_{sk_F}(\widetilde{\textbf{h}})=\langle h_i\rangle_{i\in[m]}$, $\textbf{L.D}_{sk_L}(\widehat{e_{\pi(i)}})=h_i-\textbf{L.D}_{sk_L}(\widehat{k_i})$ and $\textbf{L.D}_{sk_L}(\widehat{\eta_{\pi(i)}})=y_i\textbf{L.D}_{sk_L}(\widehat{e_{\pi(i)}})$ for $i\in[m]$.
\end{flushleft}
\begin{algorithmic}[1]
\State $\mathcal{S}$: compute $\langle h_i\rangle_{i\in[m]}\leftarrow\textbf{F.D}_{sk_F}(\widetilde{\textbf{h}})$;\Comment{Decrypt the masked evaluations}
\State $\mathcal{S}$: sample a uniform permutation $\pi\xleftarrow{\$}\mathfrak{S}_m$;
\ForAll{$i \in [m]$}
\State $\mathcal{S}$: compute $\widehat{e_{\pi(i)}}\leftarrow\textbf{L.E}_{pk_L}(h_i)-_L\widehat{k_i}$;\Comment{Remove the masks under LHE}
\State $\mathcal{S}$: compute $\widehat{\eta_{\pi(i)}}\leftarrow y_i\ltimes_L \widehat{e_{\pi(i)}}$;\Comment{Multiply by the evaluation point}
\EndFor
\State $\mathcal{S}$: send $\lbrace (\widehat{e_{\pi(i)}},\widehat{\eta_{\pi(i)}})\rbrace_{i\in[m]}$ to $\mathcal{R}$;
\State $\mathcal{R}$: \Return $\lbrace (\widehat{e_j},\widehat{\eta_j})\rbrace_{j\in[m]}$;
\end{algorithmic}
\end{algorithm}

\begin{algorithm}[htbp]
\caption{$\Union(\textbf{X},\lbrace (\widehat{e_j},\widehat{\eta_j})\rbrace_{j\in[m]},keys_\mathcal{R})$}\label{algo:Union1}
\begin{flushleft} \textbf{Input:} A set of plaintexts
$\textbf{X}=\lbrace x_j\rbrace_{j\in[n]}\subset\mathbb{M}$, a set of LHE ciphertext pairs
$\lbrace (\widehat{e_j},\widehat{\eta_j})\rbrace_{j\in[m]}\subset{\mathbb{E}_L}^2$.\\
\textbf{Output:} A set of plaintexts $\textbf{Z}\subset\mathbb{M}$, such that $\textbf{Z}=\textbf{X}\cup\lbrace \textbf{L.D}_{sk_L}(\widehat{\eta_j})\textbf{L.D}_{sk_L}(\widehat{e_j})^{-1}:\textbf{L.D}_{sk_L}\neq 0\rbrace$.\\
\end{flushleft}
\begin{algorithmic}[1]
\State
$\mathcal{R}$: compute $\textbf{Z}\leftarrow \textbf{X}$;
\ForAll{$j \in [m]$} \State
$\mathcal{R}$: compute $e_j\leftarrow
\textbf{L.D}_{sk_L}(\widehat{e_j})$;
\If{$e_j\neq 0$} \State
$\mathcal{R}$: compute $\textbf{Z}\leftarrow \textbf{Z} +\lbrace
\textbf{L.D}_{sk_L}(\widehat{\eta_j})e_j^{-1}\rbrace$;
\Comment{Add the non-roots to the set}
\EndIf
\EndFor
\State $\mathcal{R}$: \Return \textbf{Z};
\end{algorithmic}
\end{algorithm}

\begin{protocol*}[ht]
\caption{Optimal communication volume, low depth, batchable and parallelizable UPSU protocol}\label{pro:para/batch}
\begin{tabular}{l c l}
%\huge{$\mathcal{R}$}&&\huge{$\mathcal{S}$}\\
{\large$\mathcal{R}$:} $\textbf{X}=\lbrace x_i\rbrace_{i\in[n]}$ & & {\large$\mathcal{S}$:} $\textbf{Y}=\lbrace y_i\rbrace_{i\in[m]}$\\ \hline\hline
&\Setup &\\
$\lbrace pk_F,sk_L,pk_L\rbrace$&$\xleftrightarrow{\hspace{3cm}}$&$\lbrace pk_L,sk_F,pk_F\rbrace$\\
\hdashline
&\Encode & \\
&$\xleftarrow[\hspace{3cm}]{\widetilde{\textbf{y}}}$&$\widetilde{\textbf{y}}\leftarrow\textbf{F.E}_{pk_F}(\langle y_1\rangle_{i\in[m]})$\\
\hdashline
$P_{\mathcal{R}}\leftarrow\prod(Z-x_i)$
&\Reduce &\\
$k_i \xleftarrow{\$}\mathbb{M}$&&\\
$\widehat{k_i}\leftarrow\textbf{L.E}_{pk_L}( k_i )$&&\\
$\widetilde{\textbf{k}}\leftarrow\textbf{F.E}_{pk_F}(\langle k_i\rangle_{i\in[m]})$&&\\
$\widetilde{\textbf{e}}\leftarrow\textbf{F.BSMEv}(P_{\mathcal{R}},\widetilde{\textbf{y}})$&&\\
$\widetilde{\textbf{h}}\leftarrow \textbf{flood}\left(\widetilde{\textbf{k}}+_F\widetilde{\textbf{e}}\right)$&$\xrightarrow[\hspace{3cm}]{\widetilde{\textbf{h}},\ \lbrace \widehat{k_i} \rbrace_{i\in[m]}}$&\\
\hdashline
&\Map &$\langle h_i\rangle_{i\in[m]}\leftarrow\textbf{F.D}_{sk_F}(\widetilde{\textbf{h}})$\\
&&$\pi\xleftarrow{\$}\textfrak{S}_m $\\
&&$\widehat{e_{\pi(i)}}\leftarrow\widehat{h_i}-_L\widehat{k_i}$\\
&$\xleftarrow[\hspace{3cm}]{\left\{\left(\widehat{e_{\pi (i)}},\widehat{\eta_{\pi (i)}}\right)\right\}_{i\in[m]}}$&$\widehat{\eta_{\pi(i)}}\leftarrow y_i\ltimes_L\widehat{e_{\pi(i)}}$\\
\hdashline
$e_j\leftarrow\textbf{L.D}_{sk_L}\left(\widehat{e_j}\right)$
&\Union &\\
$\left\{\begin{array}{ll}
e_j=0 \Rightarrow& \perp\\
e_j\neq 0 \Rightarrow& \left\{\begin{array}{l}
\eta_j\leftarrow \textbf{L.D}_{sk_L}\left(\widehat{\eta_j}\right)\\
\textbf{X}\leftarrow\textbf{X}\cup\{ \eta_j e_j^{-1}\}\end{array}\right.
\end{array}\right.
$&&\\
\hline\hline
\textbf{Return} \textbf{X}&&\\
\end{tabular}
\end{protocol*}

%\newpage

\subsection{Timings in HElib}\label{subs:timings}
To instantiate~\cref{pro:para/batch}, we used the C++ open source library
HElib which implements the BGV cryptosystem
\cite{DBLP:journals/toct/BrakerskiGV14}; the source code is available
at~\url{https://github.com/GalanAl/CoUPSU}. We could have done our
implementations on any other library implementing exact FHE schemes
(BGV, BFV, ...), as
SEAL\footnote{\url{https://github.com/microsoft/SEAL}} and
OpenFHE\footnote{\url{https://github.com/openfheorg/openfhe-development}},
but we found more freedom in the choice of parameters with HElib. We
used a Intel(R) Xeon(R) Gold 6330 CPU @ 2.00GHz with 56 threads for our experiments. The plaintext space is the field
$\mathbb{F}_{6143^{32}}$, that contains 384 bit-length words. The
ciphertext space allows to batch up to $1024$ plaintexts together. In
the following, we consider the size of the sender's set to be a
constant $m=|\textbf{Y}|=1024$, and we want to analyze the
communication volume and the runtime of both parties when the
receiver's set size $n=|\textbf{X}|$ grows exponentially, from
$2^{10}$ to $2^{20}$ elements. We choose the other parameters of the
context in order to keep the correctness of the decryption after an
homomorphic circuit of multiplicative depth 21. Through the security
estimation given by HElib, we have $\kappa=115$. 

In our implementation, we use BGV restricted to linear operations as a LHE. Since the receiver is assumed to have significant computing power, we use parallelization on all the available threads, mainly for \textbf{F.BSMEv} in \Reduce. We ignore the communication volume and the runtime implied by \Setup as it can be done offline, and consider that both parties own a secret key and the two public keys. In the unbalanced situation, it is interesting to distinguish the runtime of both parties, presented in~\cref{fig:t(R)pro1} and~\cref{fig:t(S)pro1} respectively.

The experimental results confirm the expected asymptotic presented in~\cref{tab:costUPSU1}. In particular, the arithmetic cost for the sender is independent of the size of the receiver's set, and is low. On the receiver's side, the runtime values confirm the linear complexity in the size of its set. The runtime for the receiver is however quite large, around 5800 seconds for a set of size $2^{22}$, but the implementation could be optimized to reduce it. Also, the algorithm is highly parallelizable and the receiver is assumed to have a significant computing power, so the runtime may also be decreased by using more threads. \cref{fig:compro1} confirms that the communication volume is independent of the receiver's set size and we compare it to the values given in~\cite{Tu:2023:CCS:UPSU,Zhang:2023:Usenix:LPSU,DBLP:journals/iacr/ZhangCLPHWW24}.\footnote{In the other papers, the element bit-length is 128, but it is 384 in our implementations.} We observe that due to some optimizations, for example the leaky partitioning of~\cite{Tu:2023:CCS:UPSU}, the communication volume of previous protocols are smaller than ours when the receiver's set is small. Yet, for larger sets we obtain better results thanks to the optimal asymptotic volume of our protocol. %for small receiver's set sizes, but our protocol becomes optimal asymptotic costs becom, at one point we will always have the lowest values.
\begin{figure}[ht]
\caption{\cref{pro:para/batch} experimental results}\label{fig:all}
\begin{subfigure}{0.5\textwidth}
\caption{Receiver's runtime of~\cref{pro:para/batch}}\label{fig:t(R)pro1}
\begin{tikzpicture}
\begin{loglogaxis}[
    xlabel={Receiver's set size $\vert \textbf{X}\vert $},
    ylabel={Runtime (s), $|\textbf{Y}|=2^{10}$},
    xmin=256, xmax=8382608,
    ymin=10, ymax=10000,
    xtick={256,1024,4096,16384,65536,262144,1048576,4191304},
    ytick={10,100,1000,10000},
    legend pos=north west,
    ymajorgrids=true,
    grid style=dashed,
    width=\textwidth,
    log basis x = {2}
]
\addplot[
    color=blue,
    mark=square,
    ]
    coordinates {
    (1024,13.18)(4096,18.7)(16384,35.35)(65536,99.54)(262144,326.03)(1048576,1404.9)(4191304,5751.24)
    };
    %\legend{$d_{\max}= 21\ (\kappa=115)$}
\end{loglogaxis}
\end{tikzpicture}
\end{subfigure}
\begin{subfigure}{0.5\textwidth}
\caption{Sender's runtime of~\cref{pro:para/batch}}\label{fig:t(S)pro1}
\begin{tikzpicture}%[scale=.8]
\begin{semilogxaxis}[
    xlabel={Receiver's set size $\vert \textbf{X}\vert $},
    ylabel={Runtime (s), $|\textbf{Y}|=2^{10}$},
    xmin=256, xmax=8382608,
    ymin=0, ymax=3,
    xtick={256,1024,4096,16384,65536,262144,1048576,4191304},
    ytick={0,1,2,3,4,5},
    legend pos=north west,
    ymajorgrids=true,
    grid style=dashed,
    width=\textwidth,
    log basis x = {2}
]
\addplot[
    color=blue,
    mark=square,
    ]
    coordinates {
    (1024,2.16)(4096,2.16)(16384,2.16)(65536,2.16)(262144,2.16)(1048576,2.16)(4191304,2.16)
    };
    %\legend{$d_{\max}= 21\ (\kappa=115)$}
\end{semilogxaxis}
\end{tikzpicture}
\end{subfigure}
\begin{subfigure}{0.645\textwidth}
\caption{Communication volume of~\cref{pro:para/batch}}\label{fig:compro1}
\begin{tikzpicture}%[scale=.8]
\begin{loglogaxis}[
    xlabel={Receiver's set size $\vert \textbf{X}\vert $},
    ylabel={Communication (in MB), $|\textbf{Y}|=2^{10}$},
    xmin=256, xmax=8382608,
    ymin=0, ymax=120,
    xtick={256,1024,4096,16384,65536,262144,1048576,4191304},
    ytick={0,1,10,100},
    legend pos= outer north east,
    ymajorgrids=true,
    grid style=dashed,
    width=\textwidth,
    log basis x = {2}
]
\addplot[
    color=blue,
    mark=square,
    ]
    coordinates {
    (1024,8.39)(4096,8.39)(16384,8.39)(65536,8.39)(262144,8.39)(1048576,8.39)(4191304,8.5)
    };
    \addlegendentry{$\begin{array}{l}
    \textrm{\cref{pro:para/batch} } (\kappa=115)\\
    (\kappa =97 \textrm{ for }n=2^{22})\\
    \,
    \end{array}$}
%\addplot[
%    color=green,
%    mark=square,
%    ]
%    coordinates {
%    (1024,7.86)(4096,7.86)(16384,7.86)(65536,7.86)(262144,7.86)
%    };
%    \addlegendentry{$d_{\max}= 19\ (\kappa=125)$}
%\addplot[
%    color=yellow,
%    mark=square,
%    ]
%    coordinates {
%    (1024,7.34)(4096,7.34)(16384,7.34)(65536,7.34)
%    };
%    \addlegendentry{$d_{\max}= 17\ (\kappa=135)$}
%\addplot[
%    color=orange,
%    mark=square,
%    ]
%    coordinates {
%    (1024,6.81)(4096,6.81)(16384,6.81)
%    };
%    \addlegendentry{$d_{\max}= 15\ (\kappa=165)$}
%  \addplot[
%    color=red,
%    mark=square,
%    ]
%    coordinates {
%    (1024,6.29)(4096,6.29)
%    };
%    \addlegendentry{$d_{\max}= 13\ (\kappa=180)$}
%\addplot[
%    color=violet,
%    mark=square,
%    ]
%    coordinates {
%    (1024,5.77)
%    };
%    \addlegendentry{$d_{\max}= 11\ (\kappa=215)$}
\addplot[
    color=red,
    mark=square,
    ]
    coordinates {
    (1024,1.61)(4096,1.61)(16384,2.68)(65536,2.68)(262144,2.37)(1048576,2.77)
    };
    \addlegendentry{\cite{Tu:2023:CCS:UPSU} $(\kappa=128)$}
\addplot[
    color=brown,
    mark=square,
    ]
    coordinates {
    (1024,0.23)(4096,0.53)(16384,1.76)(65536,6.25)(262144,26.17)(1048576,104.28)
    };
    \addlegendentry{\cite{Zhang:2023:Usenix:LPSU} $(\kappa=128)$}
\addplot[
    color=orange,
    mark=square,
    ]
    coordinates {
    (262144,3.9)(1048576,6.8)(4191304,11.19)
    };
    \addlegendentry{\cite{DBLP:journals/iacr/ZhangCLPHWW24} $(\kappa=128)$}
\end{loglogaxis}
\end{tikzpicture}
\end{subfigure}
\end{figure}

%-------------------------------------------------------------------------------
%-------------------------------------------------------------------------------
%-------------------------------------------------------------------------------
\section{Asymptotic Improvement for the Receiver}\label{sec:variants}
In this section, the idea is to use efficient computer algebra
algorithms adapted to the homomorphic encryption constraints.
This allows us to improve the arithmetic cost for the receiver, making
it quasi-linear in its -- large -- set size $n$.
We are able to do this while keeping an optimal communication volume,
$O(m)$, as well as 3 rounds and a sender's arithmetic cost independent
of the receiver's set size (quasi-linear it its -- small -- set size).
The downside is that batching is now more complex to set-up and that the
multiplicative depth is doubled, making their implementations
currently less efficient on the examples of~\cref{sec:UPSU}

More precisely, we reduce the arithmetic costs of our algorithms to
fast polynomial multiplication, between two ciphertext polynomials,
one cleartext and one ciphertext, or two cleartext polynomials.
Even if fast polynomial multiplication is quasi-linear in the degree,
the context may vastly modify the actual constant in the complexity
bound.
We thus detail in the following the number of polynomial
products in each context, for $A,B$ polynomials of degrees at most~$d$:
\begin{itemize}
\item $\mathcal{M}(d)$ is a bound on the arithmetic cost of the map
  $A,B\mapsto AB$;
\item $\mathcal{M}_F(d)$ is a bound on the arithmetic cost of the map
  $\widetilde{A},\widetilde{B}\mapsto\widetilde{A}\times_F\widetilde{B}$
  (note that a homomorphic polynomial product is considered to have a multiplicative depth 1 in that paper);
\item $\mathcal{M}_L(d)$ is a bound on the arithmetic cost of the map
  $A,\widehat{B}\mapsto A\ltimes_L\widehat{B}$.
\end{itemize}
\subsection{Efficient LHE and FHE Fast Multi-point Evaluation}
We consider $P_\mathcal{R}$ and $P_\mathcal{S}$, the two polynomials
whose roots are respectively the elements of the sets \textbf{X} and
\textbf{Y}.
To avoid the set intersection, we want to discard the points of \textbf{Y}
that evaluate $P_\mathcal{R}$ to zero.
Our fist remark is that it is sufficient to use the remainder
$R=P_\mathcal{R}\mod P_\mathcal{S}$ for the evaluations.
But now this polynomial has a degree at most $m-1$ for
$m=|\textbf{Y}|$.
Our idea, in~\cref{pro:modF}, is then to homomorphically compute and
evaluate this remainder instead.
We show how to preserve a fast arithmetic complexity bound, even under
LHE or FHE.
%The resulting protocol optimizes the communication volume
%and the number of rounds, has an arithmetic cost quasi-linear in their
%sets for both parties, but does not take advantage of the batching
%procedure, is not easily parallelizable, slightly increases the depth
%compared to~\cref{pro:para/batch} and requires the LHE and FHE schemes
%to share the same field as cleartext space.
\subsubsection{Efficient Fully Homomorphic Euclidean Remainder.}\label{subs:modF}
In the context of homomorphic encryption, the impossibility of
branching, for instance, makes it more complex to design a {\em fast}
homomorphic remainder. The classical fast division algorithm computes
the quotient first by a fast modular inverse of the divisor and then
updates the remainder. The fast modular inverse is usually obtained by
a Newton-like iteration.
Here we need the division of a clear polynomial by a ciphered one.
Even then, this computation cannot be directly performed in a LHE
scheme since the divisor and the quotient, both encrypted, need to be
multiplied together.
Moreover, the need to invert the leading coefficient of the divisor
could also be an issue. We thus focus here on the case where the
divisor is monic.
We denote this algorithm with the operator $\textrm{\textbf{ mod}}_F$,
that, for $(pk_F,sk_F)\leftarrow\textbf{F.Setup}(\kappa)$, a FHE key
pair of security $\kappa$ and for two cleartext polynomials $A,B$ with
$B$ monic, do satisfy that:
\begin{equation}
\textbf{F.D}_{sk_F}\left(A \textrm{\textbf{ mod}}_F\left(\textbf{F.E}_{pk_F}\left(B\right)\right)\right)=A\mod B.
\end{equation}
\begin{proposition}\label{prop:Frem}
Let $A$ be a cleartext polynomial of degree $n$ and let
$\widetilde{B}$ be a FHE ciphertext polynomial which is an encryption
of a monic polynomial of degree $m< n$. $A\textrm{\textbf{ mod}}_F
\left(\widetilde{B}\right)$ can be computed in less than
$4\mathcal{M}_F(n-m)+O(n)$ arithmetic operations with a depth
$2\lceil\log(n-m+1)\rceil+1$. If an encryption of
$\overleftarrow{B}^{-1}\mod Z^{2^{\lceil \log m\rceil}}$ is given, the
depth is reduced to $2\lceil\log(n-m+1)-\log m\rceil+1$.
\end{proposition}
\begin{proof}
The algorithm presented in~\cref{proof:modF} is an adaptation of the fast Euclidean division algorithm, based on Newton iteration~\cite{DBLP:books/daglib/0031325}, to the FHE settings. We remark that if an encryption of $\overleftarrow{B}^{-1}\mod Z^{2^{\lceil \log m\rceil}}$ is given, where $\overleftarrow{B}$ is the reverse polynomial, the first $\lceil \log m\rceil$ steps of the Newton iterations can be skipped, reducing the multiplicative depth by $2\lceil \log m\rceil$.
\end{proof}
%\begin{remark}\label{rem:modFopti}
%In~\cref{pro:modF}, the sender gives an encryption of $\overleftarrow{P_\mathcal{S}}^{-1}\mod Z^{\lceil \log m\rceil}$, where $\overleftarrow{P_\mathcal{S}}$ is the reverse polynomial, to reduce the depth in the computation of $P_\mathcal{R}\textbf{ mod}_F(\widetilde{P_\mathcal{S}})$ by $2\lceil \log m\rceil$, without modifying the communication volume asymptotic. In fact, this extra information allow to skip the first  $\lceil \log m\rceil$ steps of the Newton iterations.
%\end{remark}
\subsubsection{Fast Linearly Homomorphic Multi-point Evaluation.}\label{subs:LMEv}
We define a linearly homomorphic multi-point
polynomial evaluation algorithm, denoted by \textbf{L.MEv}, that
homomorphically evaluates a ciphertext polynomial of degree $m-1$ in
$m$ cleartext evaluation points.
For $(pk_L,sk_L)\leftarrow\textbf{L.Setup}(\kappa)$, a LHE key pair of
security $\kappa$, for $H$ a cleartext polynomial of degree $m-1$ and
for $m$ cleartext evaluation points $\lbrace y_i\rbrace_{i\in[m]}$, if
$\lbrace\widehat{e_i}\rbrace_{i\in[m]}\leftarrow\textbf{L.MEv}\left(\textbf{L.E}_{pk_L}(H),\lbrace{y_i}\rbrace_{i\in[m]}\right)$
then, this algorithm has to satisfy, for all $i\in[m]$, that:
$\textbf{L.D}_{sk}(\widehat{e_i})=H(y_i)$.
Naively, the algorithm \textbf{L.MEv} can be implemented in $O(m^2)$
operations, evaluating homomorphically the polynomial on each point
with an adaptation of the Horner scheme.
In~\cref{proof:LMEv} we instead derive
from~\cite{DBLP:conf/issac/BostanLS03} an asymptotically fast
multi-point evaluation algorithm that satisfies:
\begin{proposition}\label{prop:Lmultev}
Let $\widehat{H}$ be an LHE ciphertext polynomial of degree $m-1$ and
$\lbrace y_i\rbrace_{i\in[m]}$ be $m$ cleartext evaluation
points. With $\frac{1}{2}\mathcal{M}(m)\log m+ \widetilde{O}(m)$
operations of precomputation on the $y_i$,
$\textbf{L.MEv}(\widehat{H},\lbrace y_i\rbrace_{i\in[m]})$ can
be computed in at most $\mathcal{M}_L(m)\log m + \widetilde{O}(m)$
operations.
\end{proposition}
\subsubsection{Protocol with $\textbf{mod}_F$ and $\textbf{L.MEv}$.}
Our second protocol makes use of two cryptosystems, a LHE and a FHE.
To ensure correctness, the schemes have to share the same cleartext
space. It is presented in~\cref{pro:modF} and
its correctness is implied by the correctness of
the encryption schemes and the fact that the remainder
$P_\mathcal{R}\mod P_\mathcal{S}$ vanishes only on the intersection
set.
\begin{protocol*}[ht]
\caption{Optimal communication volume UPSU protocol with $\textbf{mod}_F$ and $\textbf{L.MEv}$}\label{pro:modF}
\begin{tabular}{l c l}
%\huge{$\mathcal{R}$}&&\huge{$\mathcal{S}$}\\
{\large$\mathcal{R}$:} $\textbf{X}=\lbrace x_i\rbrace_{i\in[n]}$ & & {\large$\mathcal{S}$:} $\textbf{Y}=\lbrace y_i\rbrace_{i\in[m]}$\\ \hline\hline
&\Setup &\\
$\lbrace pk_F,sk_L,pk_L\rbrace$&$\xleftrightarrow{\hspace{3cm}}$&$\lbrace pk_L,sk_F,pk_F\rbrace$\\
\hdashline
&\Encode & $P_{\mathcal{S}}\leftarrow\prod(Z-y_i)$\\
&&$\widetilde{U_l}\leftarrow \overleftarrow{P_{\mathcal{S}}}^{-1}\mod Z^{2^l}$\\
&&$\widetilde{U_l}\leftarrow \textbf{F.E}_{pk_F}(U_l)$\\
&$\xleftarrow[\hspace{3cm}]{\widetilde{P_{\mathcal{S}}},\widetilde{U_l}}$&$\widetilde{P_{\mathcal{S}}}\leftarrow\textbf{F.E}_{pk_F}(P_{\mathcal{S}})$\\%$\xleftarrow[\hspace{3cm}]{\widetilde{P_{\mathcal{S}}},\ \widetilde{U_l}}$&$\widetilde{P_{\mathcal{S}}},\widetilde{U_l}\leftarrow\textbf{F.E}_{pk_F}(P_{\mathcal{S}},U_l)$
\hdashline
$P_{\mathcal{R}}\leftarrow\prod(Z-x_i)$
&\Reduce &\\
$M \xleftarrow{\$}\mathbb{M}[Z]_{m-1}$&&\\
$\widehat{M}\leftarrow\textbf{L.E}_{pk_L}(M)$&&\\
$\widetilde{M}\leftarrow\textbf{F.E}_{pk_F}(M)$&&\\
$\widetilde{R}\leftarrow P_{\mathcal{R}} \textbf{mod}_F(\widetilde{P_{\mathcal{S}}})$&&\\
$\widetilde{H}\leftarrow \widetilde{R}+_F\widetilde{M}$&$\xrightarrow[\hspace{3cm}]{\widetilde{H},\ \widehat{M}}$&\\
\hdashline
&\Map &$H\leftarrow\textbf{F.D}_{sk_F}(\widetilde{H})$\\
&&$\lbrace h_i\rbrace_{i\in[m]} \leftarrow \textbf{MEv}(H,\textbf{Y})$\\
&&$\lbrace \widehat{m_i}\rbrace_{i\in[m]} \leftarrow \textbf{L.MEv}(\widehat{M},\textbf{Y})$\\
&&$\pi\xleftarrow{\$}\textfrak{S}_m $\\
&&$\widehat{e_{\pi(i)}}\leftarrow\widehat{h_i}-_L\widehat{m_i}$\\
&$\xleftarrow[\hspace{3cm}]{\left\{\left(\widehat{e_{\pi (i)}},\widehat{\eta_{\pi (i)}}\right)\right\}_{i\in[m]}}$&$\widehat{\eta_{\pi(i)}}\leftarrow y_i\ltimes_L\widehat{e_{\pi(i)}}$\\
\hdashline
$e_j\leftarrow\textbf{L.D}_{sk_L}\left(\widehat{e_j}\right)$
&\Union &\\
$\left\{\begin{array}{ll}
e_j=0 \Rightarrow& \perp\\
e_j\neq 0 \Rightarrow& \left\{\begin{array}{l}
\eta_j\leftarrow \textbf{L.D}_{sk_L}\left(\widehat{\eta_j}\right)\\
\textbf{X}\leftarrow\textbf{X}\cup\{ \eta_j e_j^{-1}\}\end{array}\right.
\end{array}\right.
$&&\\
\hline\hline
\textbf{Return} \textbf{X}&&\\
\end{tabular}
\end{protocol*}

The full simulation proof of the security of~\cref{pro:modF} is given
in~\cref{app:secuUPSU}.
The asymptotic presented in~\cref{tab:costUPSU2} are implied
by~\cref{prop:Lmultev} and~\cref{prop:Frem}. In particular, the sender
sends $\widetilde{P_{\mathcal{S}}}$ together with an encryption of
$\overleftarrow{P}_{\mathcal{S}}^{-1}\mod Z^{2^{\lceil\log{m}\rceil}}$
in its first message, to reduce the FHE multiplicative depth of
$P_{\mathcal{R}}~\textbf{mod}_F(\widetilde{P_{\mathcal{S}}})$, without
increasing the dominant term of the communication volume.
We observe that each
party has a quasi-linear arithmetic cost in its set size, while
conserving an optimal communication volume. Also, we can remark that
this protocol has a lower multiplicative depth
than~\cref{pro:para/batch} when $\sqrt{n}<m<n$, even if it is not our
target usage case.
\begin{theorem}\label{thm:UPSU2}
\Cref{pro:modF} is a secure UPSU under the honest-but-curious
adversary model. For the receiver and the server respectively owning
sets \textbf{X} and \textbf{Y} of $n$ and $m$ cleartexts, with the
assumption that $n> m$, it computes the set union with the asymptotic
complexity bounds presented in~\cref{tab:costUPSU2}.
\begin{table}[ht]
\centering
\caption{Cost analysis of~\cref{pro:modF} for $n> m$}\label{tab:costUPSU2}
\begin{tabular}{l||c|c|c|c|}
Algorithm&Ar. Cost for $\mathcal{R}$&Ar. Cost for $\mathcal{S}$&
Comm. Vol.& Depth\\
\hline\hline
\Setup &$O(1)$&$O(1)$&$O(1)$&\\
\hdashline
\Encode &$O(1)$&$\widetilde{O}(m)$&\boldmath$O(m)$&\\
\hdashline
\Reduce &\boldmath$4\mathcal{M}_F(n)+\widetilde{O}(n)$&$O(1)$&\boldmath$O(m)$&\boldmath$\begin{array}{l}
2(\lceil\log (n-m+1)\rceil\\
-\lceil \log m \rceil) +1
\end{array}$\\
\hdashline
\Map &$O(1)$&\boldmath$\begin{array}{l}\mathcal{M}_L(m)\log m\\ + \widetilde{O}(m)\end{array}$&\boldmath$O(m)$&\\
\hdashline
\Union &$O(m)$&$O(1)$&$O(1)$&\\
\hline\hline
\textbf{Total}&$4\mathcal{M}_F(n)+\widetilde{O}(n)$&$\begin{array}{l}\mathcal{M}_L(m)\log m\\ +\widetilde{O}(m)\end{array}$&$O(m)$&$\begin{array}{l}
2(\lceil\log (n-m+1)\rceil\\
-\lceil \log m \rceil) +1
\end{array}$\\
\end{tabular}
\end{table}
\end{theorem}

%\newpage
\subsection{Fast FHE Multi-point Evaluation with Compatibility LHE-FHE}
In fact, as long as the divisor is monic, we show in the following
that it is actually possible to perform a fast multi-point evaluation, based on successive
remainders, from a FHE Euclidean remainder.
This is the idea of the variant version presented in~\cref{pro:FMEv}.
This protocol allows a quasi-linear arithmetic cost for both parties,
in their respective set size and an optimal communication volume.
It is clearer to present it when the LHE and the FHE have the same
cleartext space, but we show in~\cref{subsec:compati} that this
condition is not mandatory.
The algorithm can be parallelized but batching is more difficult, and
its depth is twice that of~\cref{pro:para/batch}. In comparison to~\cref{pro:modF}, the cost for $\mathcal{S}$ is reduced, as it does not need to perform the homomorphic multi-point evaluation, but the multiplicative depth is increased by $2\log m$.
\subsubsection{Fast Fully Homomorphic Multi-point Evaluation.}
We thus now denote by \textbf{F.MEv}, a fast FHE multi-point
evaluation algorithm, that homomorphically evaluates a cleartext
polynomial in encrypted evaluation points.
For $(pk_F,sk_F)\leftarrow\textbf{F.Setup}(\kappa)$, a FHE key pair of
security $\kappa$, for a cleartext polynomial $A$ and for $m$
cleartext evaluation points $\lbrace y_i\rbrace_{i\in[m]}$, if
$\lbrace\widetilde{e_i}\rbrace_{i\in[m]}\leftarrow\textbf{F.MEv}(A,\lbrace\textbf{F.E}_{pk_F}(y_i)\rbrace_{i\in[m]})$
then, this algorithm has to satisfy, for all $i\in[m]$, that:
$\textbf{F.D}_{sk_F}(\widetilde{e_i})=A(y_i)$.
We show in~\cref{proof:FMEv} that it is possible to adapt a Newton
iterations and a fast multi-point algorithm to the fully homomorphic
context.
This allows us to get the following proposition:
\begin{proposition}\label{prop:FMEv} Let $A$ be of degree $n$ and
  $\lbrace\widetilde{y_i}\rbrace_{i\in[m]}$ be $m<n$
  ciphertexts. $\textbf{F.MEv}(A,\lbrace
  \widetilde{y_i}\rbrace_{i\in[m]})$ can be computed in less than
  $4\mathcal{M}_F(n-m)+O(n)$ arithmetic operations with a depth
  $2(L+l)$, for $L=\lceil \log(n-m+1)\rceil$ and
  $l=\lceil\log{m}\rceil$. If an encryption of
  $\overleftarrow{\prod\limits_m(Z-y_i)}^{-1}\mod Z^{2^l}$ is given,
  the depth is reduced to $2L$.
\end{proposition}
\begin{proof}
The algorithm presented in~\cref{proof:FMEv} is a modification of the
Newton iterations and a multi-point algorithm based on successive
Euclidean remainders adapted in the FHE context. Similarly
to~\cref{pro:modF}, having an encryption of
$\overleftarrow{\prod\limits_m(Z-y_i)}^{-1}\mod Z^{2^l}$ allows to
skip $l$ steps of the Newton iterations, and reduce the multiplicative
depth by~$2l$.
\end{proof}
\subsubsection{Protocol with \textbf{F.MEv}.}
This allows us to present~\cref{pro:FMEv} whose characteristics are
given in~\cref{thm:UPSU3}.

\begin{protocol*}[ht]
\caption{Optimal communication and arithmetic cost with compatibility LHE-FHE protocol}\label{pro:FMEv}
\begin{tabular}{l c l}
%\huge{$\mathcal{R}$}&&\huge{$\mathcal{S}$}\\
{\large$\mathcal{R}$:} $\textbf{X}=\lbrace x_i\rbrace_{i\in[n]}$ & & {\large$\mathcal{S}$:} $\textbf{Y}=\lbrace y_i\rbrace_{i\in[m]}$\\ \hline\hline
&\Setup &\\
$\lbrace pk_F,sk_L,pk_L\rbrace$&$\xleftrightarrow{\hspace{3cm}}$&$\lbrace pk_L,sk_F,pk_F\rbrace$\\
\hdashline
&\Encode & \\
&&$P_{\mathcal{S}}\leftarrow\prod(Z-y_i)$\\
&&$\widetilde{U_l}\leftarrow \overleftarrow{P_{\mathcal{S}}}^{-1}\mod Z^{2^l}$\\
&&$\widetilde{U_l}\leftarrow \textbf{F.E}_{pk_F}(U_l)$\\
%&&$\widetilde{y_i}\leftarrow \textbf{F.E}_{pk_F}(y_i)$\\
&$\xleftarrow[\hspace{3cm}]{\lbrace\widetilde{y_i}\rbrace_{i\in[m]},\widetilde{U_l}}$&$\widetilde{y_i}\leftarrow \textbf{F.E}_{pk_F}(y_i)$\\%$\xleftarrow[\hspace{3cm}]{\lbrace\widetilde{y_i}\rbrace_{i\in[m]},\ \widetilde{U_l}}$&$\widetilde{P_{\mathcal{S}}},\widetilde{U_l}\leftarrow\textbf{F.E}_{pk_F}(P_{\mathcal{S}},U_l)$
\hdashline
$P_{\mathcal{R}}\leftarrow\prod(Z-x_i)$
&\Reduce &\\
$k_i \xleftarrow{\$}\mathbb{M}$&&\\
$\widehat{k_i}\leftarrow\textbf{L.E}_{pk_L}(k_i)$&&\\
$\widetilde{k_i}\leftarrow\textbf{F.E}_{pk_F}(k_i)$&&\\
$\lbrace\widetilde{e_i}\rbrace_{i\in[m]}\leftarrow \textbf{F.MEv}(P_{\mathcal{R}},\lbrace \widetilde{y_i}\rbrace_{i\in[m]})$&&\\
$\widetilde{h_i}\leftarrow \widetilde{k_i}+_F\widetilde{e_i}$&$\xrightarrow[\hspace{3cm}]{\lbrace\widetilde{h_i}\rbrace_{i\in[m]},\ \lbrace\widehat{k_i}\rbrace_{i\in[m]}}$&\\
\hdashline
&\Map &$h_i\leftarrow\textbf{F.D}_{sk_F}(\widetilde{h_i})$\\
&&$\pi\xleftarrow{\$}\textfrak{S}_m $\\
&&$\widehat{e_{\pi(i)}}\leftarrow\widehat{h_i}-_L\widehat{m_i}$\\
&$\xleftarrow[\hspace{3cm}]{\left\{\left(\widehat{e_{\pi (i)}},\widehat{\eta_{\pi (i)}}\right)\right\}_{i\in[m]}}$&$\widehat{\eta_{\pi(i)}}\leftarrow y_i\ltimes_L\widehat{e_{\pi(i)}}$\\
\hdashline
$e_j\leftarrow\textbf{L.D}_{sk_L}\left(\widehat{e_j}\right)$
&\Union &\\
$\left\{\begin{array}{ll}
e_j=0 \Rightarrow& \perp\\
e_j\neq 0 \Rightarrow& \left\{\begin{array}{l}
\eta_j\leftarrow \textbf{L.D}_{sk_L}\left(\widehat{\eta_j}\right)\\
\textbf{X}\leftarrow\textbf{X}\cup\{ \eta_j e_j^{-1}\}\end{array}\right.
\end{array}\right.
$&&\\
\hline\hline
\textbf{Return} \textbf{X}&&\\
\end{tabular}
\end{protocol*}

The correctness proof of~\cref{pro:FMEv} follows the line of that
of~\cref{pro:para/batch} and the simulator is given
in~\cref{app:secuUPSU}.
Its asymptotic cost is presented in~\cref{tab:costUPSU3} and is
implied by~\cref{prop:FMEv}. In particular, the sender sends
$\widetilde{P_{\mathcal{S}}}$ together with an encryption of
$\overleftarrow{\prod\limits_m(Z-y_i)}^{-1}\mod{Z}^{2^{\lceil\log{m}\rceil}}$
in its first message.
This reduces the FHE multiplicative depth of
$\textbf{F.MEv}(P_{\mathcal{R}},\lbrace\widetilde{y_i}\rbrace_{i\in[m]})$,
without increasing the dominant term of the communication volume.
We have proven:
\begin{theorem}\label{thm:UPSU3}
\Cref{pro:FMEv} is a secure UPSU under the honest-but-curious adversary model. For the receiver and the server respectively owning sets \textbf{X} and \textbf{Y} of $n$ and $m$ cleartexts, with the assumption that $n> m$, it computes the set union with the asymptotic complexity bounds presented in~\cref{tab:costUPSU3}.
\begin{table}[ht]
\centering
\caption{Cost analysis of~\cref{pro:FMEv} for $n> m$}\label{tab:costUPSU3}
\begin{tabular}{l||c|c|c|c|}

Algorithm&Ar. Cost for $\mathcal{R}$&Ar. Cost for $\mathcal{S}$&
Comm. Vol.& Depth\\
\hline\hline
\Setup &$O(1)$&$O(1)$&$O(1)$&\\
\hdashline
\Encode &$O(1)$&\boldmath$\widetilde{O}(m)$&\boldmath$O(m)$&\\
\hdashline
\Reduce &\boldmath$4\mathcal{M}_F(n)+O(n)$&$O(1)$&\boldmath$O(m)$&\boldmath$2\log(n-m+1)$\\
\hdashline
\Map &$O(1)$&$O(m)$&\boldmath$O(m)$&\\
\hdashline
\Union &$O(m)$&$O(1)$&$O(1)$&\\
\hline\hline
\textbf{Total}&$4\mathcal{M}_F(n)+O(n)$&$\widetilde{O}(m)$&$O(m)$&$2\log(n-m+1)$\\
\end{tabular}
\end{table}
\end{theorem}
\begin{remark}
According to~\cref{prop:FMEv}, the sender's arithmetic cost can be reduced to be linear in its set size, but this would lead to an increase of the multiplicative depth by $2\log m$.
\end{remark}

\section{Conclusion}

We have presented two new protocols for the UPSU problem, which are the
first whose computational and communication costs for the sender depend
only on the size of the sender's set. Asymptotically, we have achieved
the optimal or nearly-optimal computational and communication costs.

However, our protocols are still largely of theoretic interest and
further work remains to make them viable in practice.
Our preliminary experiments for the first protocol validate our
asymptotics and show that the sender
computation time and communication volume are competitive with other
recent UPSU work. However, we would need to scale this to larger sizes
(particularly larger \emph{receiver} set sizes) to demonstrate the
practical advantage over alternative approaches.

This scaling is hindered by the higher receiver computational cost in our first
protocol. Our second protocol has a much better (quasi-optimal)
computational cost, but due to the somewhat more sophisticated algorihms
which make use of both FHE and LHE
cryptosystems in a delicate interaction, developing a practical
implementation of the second protocol remains as future work.

%\clearpage

\bibliographystyle{acm}
\bibliography{bibfile}

\appendix
\crefalias{section}{appendix}
%-------------------------------------------------------------------------------
%-------------------------------------------------------------------------------
%-------------------------------------------------------------------------------
\section{About the Compatibility FHE-LHE}\label{subsec:compati}
In this paper, our~\cref{pro:para/batch,pro:modF,pro:FMEv} combine FHE
and LHE schemes. But, depending on the instantiation of both schemes,
it is not always straightforward to preserve the overall correctness:
denoting $\mathbb{M}_L$ and $\mathbb{M}_F$ the respective cleartext
spaces of the LHE and FHE schemes, we have presented our protocols
when the message spaces are both equal to the same finite field,
$\mathbb{M}_F=\mathbb{M}_L=\mathbb{F}_{q}$. In this case, there are no
compatibility issues and this is always possible as an FHE is by
definition also an LHE.

Another situation is when $\mathbb{M}_F=\mathbb{F}_{q}$ and
$\mathbb{M}_L=\mathbb{Z}_N$, for $q^2<N$:
for instance this is the case if BGV is used for the FHE and Paillier
cryptosystem is used for the LHE. We show next that our~\cref{pro:para/batch,pro:FMEv} can
benefit from such an instantiation.

%In summary, we need that: for \textbf{X} and \textbf{Y} the respective
%receiver's and sender's set, there are two FHE cleartexts
%$\alpha(y),\beta(y)\in\mathbb{M}_F$, linked to an
%$y\in\textbf{Y}\subset\mathbb{M}_F$, such that
%$\alpha=\beta\Leftrightarrow y\in \textbf{X}$.
%The receiver sends $\widetilde{\alpha}$, an FHE encryption of
%$\alpha$, and $\widehat{\beta}$, an LHE encryption of $\beta$, to the
%sender. The sender decrypts $\widetilde{\alpha}$, and computes
%$\widehat{e}\leftarrow \widehat{\alpha}-_L\widehat{\beta}$. Finally,
%the receiver should obtain a message allowing him to compute $y$ if
%and only if the decryption of $\widehat{e}$ is $0$. When
%$\mathbb{M}_F=\mathbb{M}_L=\mathbb{F}_{q}$, the sender computes
%$\widehat{\eta}\leftarrow y\ltimes_L \widehat{e}$ and sends
%$(\widehat{e},\widehat{\eta})$ to the receiver. After decryption, the
%later obtains $(0,0)$ if $y\in\textbf{X}$ and can compute back $y$
%with a division otherwise.
%This is a solution of~\cref{pro:para/batch,pro:modF,pro:FMEv} to end
%the protocol in one round.

Indeed, more precisely, in BGV,
$\mathbb{F}_q=\mathbb{F}_{p^k}\simeq\frac{\mathbb{F}_p[X]}{(T(X))}$
for $T(X)$ a polynomial of degree $k$, irreducible on
$\mathbb{F}_p$. Then, for a chosen $T(X)$, an element
$\alpha\in\mathbb{F}_q$ is uniquely represented as a polynomial
$\sum\limits_{i=0}^{k-1}a_iX^i\in\frac{\mathbb{F}_p[X]}{(T(X))}$,
where $a_i\in\mathbb{F}_p$ for all $i$. By viewing this polynomial as
a polynomial in $\mathbb{Z}[X]$, we can introduce the one-to-one
correspondence $\Psi$ between $\mathbb{F}_q$ and $[0,q)$ that maps
$\alpha\in\mathbb{F}_q$ to
$\sum\limits_{i=0}^{k-1}a_ip^i\in[0,q)$. With the condition $q^2<N$, an
element $\alpha\in\mathbb{F}_q$ is uniquely represented in
$\mathbb{Z}_N$ with $\Psi(\alpha)\in[0,q)\subset[0,N)$. Note that
$\Psi(0)=0$ but $\Psi$ does not preserve the arithmetic.
 In~\cref{pro:para/batch,pro:FMEv}, the receiver sends $\widetilde{\alpha}$, an FHE encryption
of $\alpha$, and $\widehat{\beta}$, an LHE encryption of
$\Psi(\beta)$, to the sender. The sender decrypts
$\widetilde{\alpha}$, and computes and sends in fact
$\widehat{e}\leftarrow\widehat{\alpha}-_L\widehat{\beta}$ as well as
$\widehat{\eta}\leftarrow\Psi(y)\ltimes_L\widehat{e}$ and
$\widehat{\nu}\leftarrow\Psi(y)\ltimes_L(\widehat{0}-_L\widehat{e})$
(where $\widehat{\alpha}$ is an encryption of $\Psi(\alpha)$).
Therefore, after decryption, if $\alpha=\beta$, all three of
$e,\eta$ and $\nu$ are $0$, as expected.
The issue is when $\alpha\neq\beta$.
There are two cases, depending whether a modular reduction was needed
in $e$ when subtracting homomorphically
$\Psi(\alpha)\in[0,q)$ to $\Psi(\beta)\in[0,q)$:
\begin{itemize}
\item If $e\in[0,q)$, then $\eta\in[0,q^2)$ and just dividing $\eta$ by $e$
over the rationals recovers $\Psi(y)=\eta{}e^{-1}=\Psi(y)ee^{-1}\in\mathbb{Q}$.
\item If now $e\in[N-q,N)$, then in fact $\nu=\Psi(y)(N-e)\in[0,q^2)$
  and we need to divide over the rationals by $N-e$ instead
  to recover $\Psi(y)=\Psi(y)(N-e)(N-e)^{-1}\in\mathbb{Q}$.
\end{itemize}
The actual case is easily decided in view if the deciphering of
$\widehat{e}$. Overall, we have shown that even in this
BGV/Paillier-like situation, the receiver can compute back
$y=\Psi^{-1}\left(\Psi(y)\right)\in\mathbb{F}_q$ if and only if
$\alpha\neq\beta$,
that is if and only if $P(y)\neq 0$.\\

For some other situations, an oblivious
transfer~\cite{DBLP:conf/soda/NaorP01} could be used
to end the protocol instead, at the cost of an increase of the number of rounds.

%If we are not in the previous situations, as long as it is possible to
%find an unique representation $\Psi(m)\in\mathbb{M}_L^l$ for each
%$m\in\mathbb{M}_F$ such that $\Psi(0)=0$ we can use an oblivious
%transfer~\cite{DBLP:conf/soda/NaorP01} to end the protocol; note that
%this solution is not optimal in terms of number of rounds. Receiving
%$\widetilde{\alpha}$, an FHE encryption of $\alpha$, and
%$\widehat{\beta}$, an LHE encryption of $\Psi(\beta)$, the sender
%sends $\widehat{e}\leftarrow \widehat{\alpha}-_L\widehat{\beta}$, for
%$\widehat{\alpha}$ an encryption of $\Psi(\alpha)$. After decryption,
%the receiver asks $y$ obliviously if and only if $e\neq 0$.

\section{Security Proofs}\label{app:secuUPSU}
We simultaneously do the simulation proof
for~\cref{pro:para/batch,pro:modF,pro:FMEv}, as their constructions
are close.
We also explicit the proof in the situation were the plaintext spaces for FHE
and LHE are the same finite field $\mathbb{M}:=\mathbb{F}_q$
(in the BGV/Paillier-like situation, otherwise, the analysis
of~\cref{subsec:compati} shows that the security is preserved).

We assume that both the FHE and LHE schemes used are IND-CPA secure.
In the following, the receiver $\mathcal{R}$ and the sender
$\mathcal{S}$ are respectively called party $1$ and $2$, while the
UPSU protocol is denoted $\Pi$.
The protocol has $3$ rounds: $\mathcal{R}$ receives $2$ messages $M_1$
and $M_3$ while $\mathcal{S}$ receives only $M_2$. The semantic
functionality is
$f:\mathcal{P}(\mathbb{M})\times\mathcal{P}(\mathbb{M})\rightarrow
\mathcal{P}(\mathbb{M})\times\mathcal{P}(\mathbb{M})$ where
$\mathcal{P}$ denotes the power set and $\mathbb{M}_F$ is the FHE
cleartext space. 

The ideal output-pair is
$f(\textbf{X},\textbf{Y})=(f_1(\textbf{X},\textbf{Y}),f_2(\textbf{X},\textbf{Y}))=((\textbf{X}\cup\textbf{Y}),\emptyset)$.
As the protocol is parametrized with (upper bounds on) the set sizes
$|\textbf{X}|$ and $|\textbf{Y}|$, those are implicitly revealed to
both parties. The following views are reduced to the minimal set that
could trivially imply the real view; for example, if the real view
have a clear polynomial $R$, a key $pk_F$ and a ciphertext
$\textbf{F.E}_{pk_F}(R)$, we omit $\textbf{F.E}_{pk}(R)$ in the view,
because if we can simulate both $R$ and $pk_F$, it is trivial to
simulate $\textbf{F.E}_{pk}(R)$. As we are proving the security of
three protocols, we will use the notation
$\left[\begin{array}{l}a\\b\\c\end{array}\right.$ to describe the
view, where $a$ is related to~\cref{pro:para/batch}, $b$
to~\cref{pro:modF} and $c$ to~\cref{pro:FMEv}. 

In the following, we
consider that the \Setup algorithm has already been completed:
it is just an exchange of two public keys $pk_L$ and $pk_F$,
respectively for the LHE and FHE scheme.
Now, the views and the outputs of each parties are:
\begin{itemize}
\item $\textbf{view}_1^\Pi(\textbf{X},\textbf{Y})=(\textbf{X},C_1,M_1,M_3)$ such that:
\begin{equation*}
C_1=\bigg\lbrace pk_L,sk_L,pk_F,\left[\begin{array}{l}\lbrace k_i\rbrace_{i\in[m]}\\M\\\lbrace k_i\rbrace_{i\in[m]}\end{array}\right.,\left[\begin{array}{l}\widetilde{\textbf{e}}\\ \widetilde{R} \\\lbrace \widetilde{e_i}\rbrace_{i\in[m]}\end{array}\right.,\lbrace e_{\pi(i)},\eta_{\pi(i)}\rbrace_{i\in[m]}\bigg\rbrace,
\end{equation*}
where $\left[\begin{array}{l}\widetilde{\textbf{e}}\\ \widetilde{R} \\\lbrace \widetilde{e_i}\rbrace_{i\in[m]}\end{array}\right.\leftarrow\left[\begin{array}{l}\textbf{F.BSMEv}(P_{\mathcal{R}},\widetilde{\textbf{y}})\\ P_{\mathcal{R}}\textbf{ mod}_F\widetilde{P_{\mathcal{S}}} \\\textbf{F.MEv}(P_\mathcal{R},\lbrace \widetilde{y_i}\rbrace_{i\in[m]})\end{array}\right.$, and for all $j\in[m]$, $e_j\leftarrow\textbf{L.D}_{sk_L}(\widehat{e_j})$, $\eta_j\leftarrow\textbf{L.D}_{sk_L}(\widehat{\eta_j})$. The content of the first message is:
\begin{equation*} M_1=\left\lbrace
\left[\begin{array}{l}\widetilde{\textbf{y}}\\\widetilde{P_\mathcal{S}},\widetilde{U_l}\\\lbrace \widetilde{y_i}\rbrace_{i\in[m]},\widetilde{U_l}\end{array}\right.\right\rbrace,
\end{equation*}
while the content of the second message is
\begin{equation*}
M_3=\left\lbrace \lbrace \widehat{e_{\pi(i)}},\widehat{\eta_{\pi(i)}}\rbrace_{i\in[m]}\right\rbrace,
\end{equation*}
for $\widehat{e_{\pi(i)}}\leftarrow\left[\begin{array}{l}\textbf{L.E}_{pk_L}(P_\mathcal{R}(y_i))\\ \textbf{L.E}_{pk_L}((P_\mathcal{R} \mod P_\mathcal{S})(y_i))\\ \textbf{L.E}_{pk_L}(P_\mathcal{R}(y_i))\end{array}\right.$ and $\widehat{\eta_{\pi(i)}}\leftarrow y_i\ltimes_L\widehat{e_{\pi(i)}}$.
\item $\textbf{view}_2^\Pi(\textbf{X},\textbf{Y})=\left(\textbf{Y},C_2,M_2\right)$ such that:
\begin{align*} C_2&=\left\lbrace
pk_F,sk_F,pk_L,\left[\begin{array}{l}\langle h_i\rangle_{i\in[m]}\\\lbrace H(y_i)\rbrace_{i\in[m]}\\\lbrace h_i\rbrace_{i\in[m]}\end{array}\right.,\lbrace \widehat{e_i}\rbrace_{i\in[m]},\pi\right\rbrace,\\
M_2&=\left\lbrace \left[\begin{array}{l}\widetilde{\textbf{h}}\\\widetilde{H}\\\lbrace \widetilde{h_i}\rbrace_{i\in[m]}\end{array}\right.,\left[\begin{array}{l}\lbrace \widehat{k_i}\rbrace_{i\in[m]}\\\widehat{M}\\\lbrace\widehat{k_i}\rbrace_{i\in[m]}\end{array}\right. \right\rbrace.
\end{align*}
\item $\textbf{output}_1^\Pi(\textbf{X},\textbf{Y})=\left(\textbf{X}\cup\left\lbrace y_{i}\in\textbf{Y}\vert P_{\mathcal{R}}(y_i)\neq 0\right\rbrace\right)$
\item $\textbf{output}_2^\Pi(\textbf{X},\textbf{Y})=\emptyset$
\end{itemize}

On the side of $\mathcal{S}$, a probabilistic polynomial-time algorithm $S_2$, taking as input the set $\textbf{Y}$, simulates $\textbf{view}_2^\Pi(\textbf{X},\textbf{Y})$ with the following tuple.
\begin{align*}
S_2(\textbf{Y},\emptyset)=\bigg(& \textbf{Y},\left\lbrace pk_F,sk_F,pk_L,\left[\begin{array}{l}\langle r_i\rangle_{i\in[m]}\\\lbrace R(y_i)\rbrace_{i\in[m]}\\\lbrace r_i\rbrace_{i\in[m]}\end{array}\right.,\left[\begin{array}{l}\lbrace \widehat{r_i}-_L\widehat{r_i'}\rbrace_{i\in[m]}\\\lbrace \widehat{R(y_i)}-_L\widehat{R'(y_i)}\rbrace_{i\in[m]}\\\lbrace \widehat{r_i}-_L\widehat{r_i'}\rbrace_{i\in[m]}\end{array}\right.,\pi' \right\rbrace,\\
& \left\lbrace \left[\begin{array}{l}\textbf{flood}(\textbf{F.E}_{pk_F}(\langle r_i\rangle_{i\in[m]}))\\\textbf{F.E}_{pk_F}(R)\\\lbrace \textbf{F.E}_{pk_F}(r_i)\rbrace_{i\in[m]}\end{array}\right. ,\left[\begin{array}{l}\lbrace \widehat{r_i'}\rbrace_{i\in[m]}\\\widehat{R'}\\\lbrace\widehat{r_i'}\rbrace_{i\in[m]}\end{array}\right. \right\rbrace \bigg),
\end{align*}
where $m=\vert \textbf{Y}\vert$, $pk_F,sk_F,pk_L$ are obtained from the \Setup algorithm, $r_i,r_i'\xleftarrow{\$}\mathbb{M}$ for all $i\in[m]$, $R,R'\xleftarrow{\$}\mathbb{M}[X]$ of degrees $m-1$, for $\widehat{r_i},\widehat{r_i'},\widehat{R},\widehat{R'}$ their LHE encryptions with $pk_L$, and $\pi'\xleftarrow{\$}\mathfrak{S}_{m}$.

In the protocol, $\lbrace h_i\rbrace_{i\in[m]}$ (resp. $H$) are masked
values and are thus indistinguishable from the random $\lbrace
r_i\rbrace_{i\in[m]}$ (resp. $R$).
Note than in~\cref{pro:para/batch}, the $\lbrace h_i\rbrace_{i\in[m]}$
are sent encrypted under FHE, and $\mathcal{S}$ has the decryption
key.
The algorithm \textbf{flood} then ensures privacy (and circuit
privacy):
both the simulated and the real
ciphertexts have the same lowest possible size and their noise are
flooded.
The $\lbrace k_i\rbrace_{i\in[m]}$ (resp. $M$) and $\pi$ are random
values, thus easily simulated with random values $\lbrace
r_i'\rbrace_{i\in[m]}$ (resp. $R'$) and $\pi'$.
By performing the operations from the protocol to the simulated
values, we obtain, for every subsets $\textbf{X},\textbf{Y}\subset\mathbb{M}$:
\begin{equation*} \lbrace
S_2(\textbf{Y},\emptyset),((\textbf{X}\cup\textbf{Y}),\emptyset)\rbrace\overset{c}{\equiv}\lbrace
\textbf{view}_2^\Pi(\textbf{X},\textbf{Y}),\textbf{output}^\Pi(\textbf{X},\textbf{Y})\rbrace.
\end{equation*}

On the side of $\mathcal{R}$, a probabilistic polynomial-time
algorithm $S_1$ taking as input the set $\textbf{X}$ and
$(\textbf{X}\cup\textbf{Y})$ first computes a set $\textbf{Z}=\lbrace
z_i\rbrace_{i\in[m]}\subset\mathbb{M}$, with $m=|\textbf{Y}|$, as
following. Exactly
$|\textbf{Y}|-(|\textbf{X}\cup\textbf{Y}|-|\textbf{X}|)=\vert
\textbf{X}\cap\textbf{Y}\vert$ of the $z_i$ (uniformly distributed)
are randomly picked in $\textbf{X}$, and the others
($|\textbf{X}\cup\textbf{Y}|-|\textbf{X}|$) are the cleartexts in
$\textbf{X}\cup\textbf{Y}\setminus \textbf{X}$. The set $\textbf{Z}$
is indistinguishable from the set $\textbf{Y}$. Indeed, the two sets
have the same size, and
$\textbf{X}\cup\textbf{Y}=\textbf{X}\cup\textbf{Z}$ by construction;
in particular,
$|\textbf{X}\cap\textbf{Y}|=|\textbf{X}\cap\textbf{Z}|$. All, the
information given by the sender to the receiver on the set
$\textbf{X}\cap\textbf{Y}$ are sent encrypted under FHE in the
protocol, and the receiver does not have the decryption key.
This implies, under the IND-CPA security of the scheme,
that taking any subset of \textbf{X} with size
$|\textbf{X}\cap\textbf{Y}|$ can lead to an indistinguishable
simulation.
Then, as the sender's input set has been properly simulated, $S_1$
simulates $\textbf{view}_1^\Pi(\textbf{X},\textbf{Y})$ this way:
\begin{equation*}
S_1(\textbf{X},(\textbf{X}\cup\textbf{Y}))=\textbf{view}_1^\Pi(\textbf{X},\textbf{Z})
\end{equation*}
Overall, we obtain for every subsets $\textbf{X},\textbf{Y}\subset\mathbb{M}$:
\begin{equation*} \lbrace
S_1(\textbf{X},(\textbf{X}\cup\textbf{Y})),((\textbf{X}\cup\textbf{Y}),\emptyset)
\rbrace\overset{c}{\equiv}\lbrace
\textbf{view}_1^\Pi(\textbf{X},\textbf{Y}),\textbf{output}^\Pi(\textbf{X},\textbf{Y})\rbrace.
\end{equation*}
\section{Efficient Homomorphic Algorithm Constructions}
In the following algorithms, for a polynomial $A=\sum_{i=0}^d a_iZ^i$,
we denote by $[A]_l^L:=\sum_{i=l}^La_iZ^{i-l}$,
and $\overleftarrow{A}:=\sum_{i=0}^d a_{d-i}Z^i$.
\subsection{FHE Euclidean Remainder}\label{proof:modF}
We recall the Newton-iteration-based algorithm for polynomial Euclidean
division. We present the fast version of~\cite{DBLP:journals/aaecc/HanrotQZ04}.
The remainder $R$ in the division of $A$ by a monic $B$, of respective degrees $n$ and $m<n$,
is the unique polynomial satisfying $A=BQ+R$ with $\deg(R) < m$. This
implies, by reversing the polynomial coefficients,
$\overleftarrow{A}=\overleftarrow{Q}\overleftarrow{B}+Z^{n-m+1}\overleftarrow{R}$, whence
\begin{equation}
\overleftarrow{Q}=\overleftarrow{A}\overleftarrow{B}^{-1}\mod
Z^{n-m+1}.
\end{equation}
The goal is to homomorphically compute the inverse
of $\overleftarrow{B}$ modulo $Z^{n-m+1}$, using
Newton iteration. Let $\widetilde{A}$ and $\widetilde{B}$ be encryptions
of $A$ and $B$, and $\widetilde{1}$ be an encryption of $1$ with the same public key.
The algorithm requires first to compute the coefficient $\widetilde{U_L}$, for $L=\lceil \log (n-m+1)\rceil -1$, of the sequence $(\widetilde{U})$:
\begin{equation}\label{eq:F.Newt}
(\widetilde{U})=\left \{
\begin{array}{lcl} \widetilde{U_0}=\widetilde{1}\\ \widetilde{U_{k+1}}=\widetilde{U_k}\times_F
\left(\widetilde{1} -_F \left[\overleftarrow{\widetilde{B}}\times_F
\widetilde{U_k}\right]_{2^k}^{2^{k+1}-1}Z^{2^k}\right)\mod Z^{2^{k+1}}
\end{array} \right.
\end{equation} Now, instead of computing the last step of the
sequence that would give us homomorphically the inverse polynomial of
$\overleftarrow{B}\mod Z^{n-m+1}$, we directly compute homomorphically the quotient, using $\widetilde{U_L}$.
\begin{align} \widetilde{S} =&\ A\ltimes_F \widetilde{U}_L \mod Z^{n-m+1}\text{, and}\\
    \widetilde{T}=&\ \left[\overleftarrow{\widetilde{B}}\times_F\widetilde{U}_L\right]_{2^L}^{2^{L+1}-1}\times_F
\left[\widetilde{S}\right]_{0}^{n-m-2^L}\mod Z^{n-m+1-2^L}.
\end{align} Then
$\overleftarrow{\widetilde{Q}}:=\widetilde{S}+_F\widetilde{T}Z^{2^L}$ is an encryption of $\overleftarrow{Q}$, the reverse
quotient. Finally, we compute
\begin{equation} \widetilde{R}=A-_F \widetilde{Q} \times_F \widetilde{B} \mod Z^m
\end{equation} to get an encryption of the remainder
$R$. Using the fact that $\mathcal{M}_F(2d)\le 2\mathcal{M}_F(d)$,
we can bound the number of
arithmetic operations done with that algorithm with at most
$4\mathcal{M}_F(n-m)+O(n)$. Also, each step of the sequence $(\widetilde{U})$ requires two homomorphic products which means that computing $\widetilde{U_L}$ has depth $2L$. To obtain $\widetilde{R}$, it requires 3 more products, so in total, this algorithm has a depth $2\lceil\log(n-m+1)\rceil+1$.
\subsection{LHE multi-point evaluation}\label{proof:LMEv}
We adapt the algorithm presented in
\cite{DBLP:conf/issac/BostanLS03} to the LHE context. Let
$\widehat{H}=\sum\limits_{i=0}^{m-1}\widehat{h_i}Z^i$, $H \leftarrow \textbf{L.D}_{sk}(\widehat{H})$.
and $y_1$, \dots, $y_m\in\mathbb{M}$.
For the sake of clarity we assume that $m$ is a power of two, but it is
not mandatory in practice. The first step of the algorithm consists in
computing the following polynomials in clear, for $k=0$, \dots,
$\log m$ and $i=1$, \dots, $2^k$:
\begin{equation}\label{eq:defsubpolys}
P_{\left(\frac{i}{2^k}\right)}:=\prod_{j\in\left\lbrace
\frac{i-1}{2^k}m+1,\dots,\frac{i}{2^k}m\right\rbrace}(Z-y_j)
\end{equation}
These polynomials can be computed using a product tree
in $\frac{1}{2}\mathcal{M}(m)\log m+\widetilde{O}(m)$
arithmetic operations. Note that these
polynomials can be precomputed if the evaluation
points are known in advance.

The algorithm requires then to compute
the polynomials
\begin{align}
    B:=& \overleftarrow{P}_{\left(\frac{1}{1}\right)}^{-1}\mod Z^m\text{, and}\\
    \widehat{A}:=& \left[\overleftarrow{B}\ltimes_L \widehat{H}\right]_{m-1}^{2m-1}.
\end{align}
Let $\widehat{A}_{\left(\frac{1}{1}\right)}:=\overleftarrow{\widehat{A}}$. The
last step of the algorithm consists in the computation
for $k=1$, \dots, $\log m$ and $i=1$, \dots, $2^k$
of the encrypted polynomials
\begin{equation}
\widehat{A}_{\left(\frac{i}{2^k}\right)}=\left[\overleftarrow{P}_{\left(\frac{i-(-1)^{(i\bmod
2)}}{2^k}\right)}\ltimes_L\widehat{A}_{\left(\frac{\lceil
i/2\rceil}{2^{k-1}}\right)}\right]_{\frac{m}{2^k}}^{\frac{m}{2^{k-1}}}.
\end{equation}
According to the correctness of the algorithm presented
in \cite{DBLP:conf/issac/BostanLS03}, $\widehat{A}_{\left(\frac{i}{m}\right)}$ is an
encryption of $H(y_i)$ for $1\leq i\le m$. The final computation of the
polynomials $\widehat{A}_{\left(\frac{i}{2^k}\right)}$ requires $\mathcal{M}_L(m)\log m +
\widetilde{O}(m)$ arithmetic operations, and this dominates the cost.
\subsection{FHE multi-point evaluation}\label{proof:FMEv}
Our goal is to evaluate a polynomial $A$ of degree $n$ in $m$ evaluation points $\lbrace y_1,...,y_m\rbrace$ homomorphically. To ease the description of the algorithm, we will assume that $m=2^l$. We will consider the polynomials $P_{\left(\frac{i}{2^k}\right)}$ presented in \eqref{eq:defsubpolys}. If the $\lbrace y_i\rbrace_{i\in[m]}$ are encrypted in a FHE scheme, one can compute those polynomials homomorphically and obtain the following sequence $(\widetilde{P})$:
\begin{equation}
(\widetilde{P})=\left\{\begin{array}{lll}
\widetilde{P_{\left(\frac{i}{m}\right)}}&=Z-_F \widetilde{y_i};&i\in [m]\\
\widetilde{P_{\left(\frac{i}{2^{k-1}}\right)}}&=\widetilde{P_{\left(\frac{2i-1}{2^{k}}\right)}}\times_F\widetilde{P_{\left(\frac{2i}{2^{k}}\right)}};&k\in [l],\ i\in [2^k]
\end{array}\right.
\end{equation}
The computation requires less than
$\frac{1}{2}l\mathcal{M}_F(m)+\widetilde{O}(m)$ arithmetic operations
with a depth $l$. Through an adaptation of the Newton iterations,
given the sequence $(\widetilde{P})$ and the encrypted $\lbrace
y_i\rbrace_{i\in[m]}$, one can compute the following sequence
$(\widetilde{V})$, that gives encryption of the set
$\left\{\overleftarrow{P}_{\left(\frac{i}{2^{l-k}}\right)}^{-1}\mod
  Z^{2^{k}}\right\}_{i\in[2^k]}$, for $k\in[0,l]$, as following:
\begin{equation*}
(\widetilde{V})=\left\{\begin{array}{lll}
\widetilde{V_0}^{\left( i\right)}&=\widetilde{1};&i\in[m]\\
\widetilde{V_1}^{\left( i\right)}&=\widetilde{1}+_F(\widetilde{y_{2i-1}}+\widetilde{y_{2i}})Z;&i\in[2^{l-1}]\\
\widetilde{V_{k+1}^{\left( i\right)}}&= \widetilde{K_0}-Z^{2^k}\left[\left[\widetilde{K_0}\right]_0^{2^k-1}\times_F(\widetilde{K_1}+_F\widetilde{K_2})\right]_0^{2^k-1}\mod Z^{2^{k+1}};&i\in[2^{l-(k+1)}]
\end{array}\right.
\end{equation*}
where, for the computation of $\widetilde{V_{k+1}^{\left( i\right)}}$, we have
\begin{align*}
\widetilde{K_0}&=\widetilde{V_{k}^{\left( 2i-1\right)}}\widetilde{V_{k}^{\left( 2i\right)}},\\
\widetilde{K_1}&=\left[ \overleftarrow{\widetilde{P_{\left(\frac{2i-1}{2^{l-k}}\right)}}}\times_F\widetilde{V_{k}^{\left( 2i-1\right)}}\right]_{2^k}^{2^{k+1}-1},\\
\widetilde{K_2}&=\left[ \overleftarrow{\widetilde{P_{\left(\frac{2i}{2^{l-k}}\right)}}}\times_F\widetilde{V_{k}^{\left( 2i\right)}}\right]_{2^k}^{2^{k+1}-1}.
\end{align*}
We remark that for all $k\in[0,l]$ and $i\in[2^{l-k}]$, $\widetilde{V_k^{\left( i\right)}}$ is an encryption of $\overleftarrow{P}_{\left(\frac{i}{2^{l-k}}\right)}^{-1}\mod Z^{2^k}$. It requires less than $2(l-2)\mathcal{M}_F(m)+\widetilde{O}(m)$ arithmetic operations. Also, this algorithm has a depth $2(l-1)$. With another sequence of Newton iteration, we want to obtain homomorphically $\overleftarrow{P}_{\left(\frac{1}{1}\right)}^{-1}\mod Z^{n-m+1}$ in order to perform the Euclidean division of $A$ by $P_{\left(\frac{1}{1}\right)}=\prod\limits_{i=1}^m (Z-y_i)$. In fact, we are computing the sequence $(\widetilde{U})$ from \eqref{eq:F.Newt}:
\begin{equation*}
(\widetilde{U})=\left \{
\begin{array}{lcl} \widetilde{U_l}=\widetilde{V_{l}^{\left( 1\right)}}\\ \widetilde{U_{k+1}}=\widetilde{U_k}\times_F
\left(\widetilde{1} -_F \left[\overleftarrow{\widetilde{P_{\left(\frac{1}{1}\right)}}}\times_F
\widetilde{U_k}\right]_{2^k}^{2^{k+1}-1}Z^{2^k}\right)\mod Z^{2^{k+1}}
\end{array} \right.
\end{equation*}
As explained in~\cref{proof:modF}, by denoting $L=\lceil\log (n-m+1)\rceil -1$ it requires less than $2\mathcal{M}_F(n-m)+O(n)$ arithmetic operations to obtain homomorphically $\widetilde{U_L}$, with a depth $2(L-l)$. With less than $2\mathcal{M}_F(n-m)+O(n)$ arithmetic operations and a depth $3$, we compute homomorphically the remainder of $A$ divided by $P_{\left(\frac{1}{1}\right)}$, that we denote $\widetilde{R_l}$. We are now applying a multi-point evaluation algorithm different from the one presented in~\cref{proof:LMEv}, which consists in successive Euclidean remainders in the $P_{\left(\frac{i}{2^k}\right)}$. The previously computed $\widetilde{P_{\left(\frac{i}{2^k}\right)}}$ and $\widetilde{V_k^{(i)}}$ will help us to do this algorithm homomorphically through the following sequence. For $k\in[0,l-1]$:
\begin{equation}
(\widetilde{R})=\left\{\begin{array}{lll}
\widetilde{R_l^{(1)}}&=\widetilde{R_l}(\ =:A\textbf{ mod}_F \widetilde{P_{\left(\frac{1}{1}\right)}});&\\
\widetilde{R_k^{(2i-1)}}&=\widetilde{R_{k+1}^{(i)}}-_F \widetilde{P_{\left(\frac{2i-1}{2^{l-k}}\right)}}\overleftarrow{\left[\widetilde{V_k^{(2i-1)}} \overleftarrow{\widetilde{R_{k+1}^{(i)}}}\right]_0^{2^k-1}}\mod Z^{2^k};&i\in[2^{l-k}]\\
\widetilde{R_k^{(2i)}}&=\widetilde{R_{k+1}^{(i)}}-_F \widetilde{P_{\left(\frac{2i}{2^{l-k}}\right)}}\overleftarrow{\left[\widetilde{V_k^{(2i)}} \overleftarrow{\widetilde{R_{k+1}^{(i)}}}\right]_0^{2^k-1}}\mod Z^{2^k};&i\in[2^{l-k}]
\end{array}\right.
\end{equation}
The sequence $(\widetilde{R})$ satisfies the following correctness, assuming $sk_F$ is the decryption key, $\forall k\in[0,l-1],\forall i\in[2^{l-k}]$ :
\begin{align}
\textbf{F.D}_{sk_F}\left(\widetilde{R_{k+1}^{(i)}}\right)&\mod P_{\left(\frac{2i-1}{2^{l-k}}\right)} = \textbf{F.D}_{sk_F}\left(\widetilde{R_{k}^{(2i-1)}}\right),\\
\textbf{F.D}_{sk_F}\left(\widetilde{R_{k+1}^{(i)}}\right)&\mod P_{\left(\frac{2i}{2^{l-k}}\right)} = \textbf{F.D}_{sk_F}\left(\widetilde{R_{k}^{(2i)}}\right).
\end{align}
In particular, we have :
\begin{equation}
\textbf{F.D}_{sk_F}\left(\widetilde{R_0^{(i)}}\right)=A(y_i)
\end{equation}
Finally, computing homomorphically all the $\left\{\widetilde{R_0^{(i)}}\right\}_{i\in[m]}$, given $\widetilde{R_l}$, all the $\widetilde{P_{\left(\frac{i}{2^k}\right)}}$ and all the $\widetilde{V_k^{(i)}}$, requires less than $2l\mathcal{M}_F(m) +\widetilde{O}(m)$ arithmetic operations. Also, this algorithm has a depth $2l$. In total, for $n>m$, less than $4\mathcal{M}_F(n-m)+O(n)$ arithmetic operations and a depth $2(L+l+1)=2(\lceil \log(n-m+1)\rceil + \lceil \log m\rceil)$ are required.

\end{document}